\documentclass[11pt,a4paper]{article}

\usepackage[margin=1in]{geometry} 
\usepackage{graphics,color}
\usepackage{epsfig}
\usepackage{amsthm,amsfonts}
\usepackage{mathrsfs}
\usepackage{algorithm,algorithmic}
\usepackage{subfigure}
\usepackage{hyperref}	
\usepackage{complexity}
\usepackage{amssymb}
\usepackage{amsmath}
\usepackage{xcolor}
\usepackage{todonotes}
\usepackage{graphicx}
\usepackage{caption}
\usepackage{tikz}

\newcommand{\mc}{\mathcal}

\newcommand{\Oh}{\mathcal{O}}

\newcommand{\allsets}{[V(G)]^{\leq k}}

\newtheorem{proposition}{Proposition}
\newtheorem{lemma}{Lemma}
\newtheorem{theorem}{Theorem}
\newtheorem{corollary}{Corollary}

\newtheorem{claim}{Claim}

\title{The complexity of independent set reconfiguration \\ on bipartite graphs}

\author{
Daniel Lokshtanov
\thanks{Department of Informatics, University of Bergen, Norway, e-mail: \texttt{daniello@ii.uib.no}}
\and 
Amer E. Mouawad
\thanks{Department of Informatics, University of Bergen, Norway, e-mail: \texttt{a.mouawad@uib.no}}
}

\date{}

\begin{document}
\maketitle
\thispagestyle{empty}

\begin{abstract} 
We settle the complexity of the {\sc Independent Set Reconfiguration} problem on 
bipartite graphs under all three commonly studied reconfiguration models. 
We show that under the token jumping or token addition/removal model the problem is \NP-complete. 
For the token sliding model, we show that the problem remains \PSPACE-complete. 
\end{abstract}

\section{Introduction}\label{sec-intro}
Many real-world problems present themselves in the following form:
given the description of a system state and the description of a state we would ``prefer'' the system to be in, is it possible to transform the system 
from its current state into the desired one without ``breaking'' the system in the process? Such questions, with some generalizations and specializations, have  
received a substantial amount of attention under the so-called \emph{reconfiguration framework}~\cite{DBLP:journals/tcs/BrewsterMMN16,MouawadN0SS17,H13,Wrochna15}. 
Historically, the study of reconfiguration questions predates the field of computer science, as 
many classic one-player games can be formulated as reachability questions~\cite{JS79,KPS08}, e.g., the $15$-puzzle and Rubik's cube. 
More recently, reconfiguration problems have emerged from computational problems in 
different areas such as graph theory~\cite{CHJ08,IDHPSUU11,IKD12}, constraint satisfaction~\cite{GKMP09,DBLP:conf/icalp/MouawadNPR15}, 
computational geometry~\cite{DBLP:books/daglib/0019278,DBLP:conf/stacs/KanjX15,DBLP:journals/comgeo/LubiwP15}, and even quantum complexity theory~\cite{DBLP:conf/icalp/GharibianS15}. 

In this work, we focus on the reconfiguration of independent sets and vertex covers of bipartite graphs.  
We view an independent set as a collection of tokens placed on the vertices of a graph such that no two tokens are adjacent. 
This gives rise to three natural adjacency relations between independent sets (or token configurations), also called \emph{reconfiguration steps}.   
In the \emph{token addition/removal} (TAR) model, first introduced by Ito et al.~\cite{IDHPSUU11}, we are allowed to either add or remove one token 
at a time as long as there are at least $k$ (non-adjacent) tokens on the graph at all times. 
In the \emph{token jumping} (TJ) model, introduced by Kami\'{n}ski et al.~\cite{KMM12}, a single reconfiguration step consists of first removing 
a token on some vertex $u$ and then immediately adding it back on any other vertex $v$, as long as no two tokens become adjacent. 
The token is said to \emph{jump} from vertex $u$ to vertex $v$. 
Finally, in the \emph{token sliding} (TS) model, introduced by Hearn and Demaine~\cite{DBLP:journals/tcs/HearnD05}, two independent sets are adjacent if one can 
be obtained from the other by a token jump from vertex $u$ to vertex $v$ with the additional requirement of $uv$ being an edge of the graph. 
The token is then said to \emph{slide} from vertex $u$ to vertex $v$ along the edge $uv$. Note that, in both 
the TJ and TS models, the size of independent sets is fixed, while the TAR model only enforces a lower bound.  
Generally speaking, in the {\sc $\mc{M}$-Independent Set Reconfiguration} ($\mc{M}$-ISR) problem, where $\mc{M} \in \{\text{TAR},\text{TJ},\text{TS}\}$, 
we are given a graph $G$ and two independent sets $I$ and $J$ of $G$. The goal is to determine 
whether there exists a sequence of reconfiguration steps -- a \emph{reconfiguration sequence} -- that transforms $I$ into $J$ (where the reconfiguration step depends on the model). 
$\mc{M}$-ISR has been extensively studied under the reconfiguration framework, albeit under 
different names~\cite{DBLP:journals/corr/BonamyB16,DBLP:conf/swat/BonsmaKW14,DBLP:conf/isaac/DemaineDFHIOOUY14,DBLP:conf/isaac/Fox-EpsteinHOU15, 
DBLP:conf/tamc/ItoKOSUY14,DBLP:journals/ieicet/ItoNZ16,KMM12,DBLP:conf/wads/LokshtanovMPRS15,DBLP:conf/isaac/MouawadNR14}.
It is known that the problem is \PSPACE-complete for all three models, even on restricted graph classes such as  
graphs of bounded bandwidth/pathwidth~\cite{DBLP:journals/corr/Wrochna14} and planar graphs~\cite{DBLP:journals/tcs/HearnD05}. 
A popular open question related to $\mc{M}$-ISR is whether the problem becomes polynomial-time solvable 
on bipartite graphs~\cite{DBLP:conf/isaac/Fox-EpsteinHOU15,presentation1,DBLP:journals/corr/HoangU17,DBLP:conf/isaac/MouawadNR14}. 
A few positive results for subclasses of bipartite graphs are known. For instance, it was shown 
by Demaine et al.~\cite{DBLP:conf/isaac/DemaineDFHIOOUY14} that TS-ISR can be solved in polynomial time on trees. 
Fox-Epstein et al.~\cite{DBLP:conf/isaac/Fox-EpsteinHOU15} gave polynomial-time algorithms 
for solving TS-ISR on bipartite permutation and bipartite distance-hereditary graphs, and conjectured that the problem remains polynomial-time solvable on bipartite graphs. 
Mouawad et al.~\cite{DBLP:conf/isaac/MouawadNR14} studied the shortest path variant of TAR-ISR, where we seek a shortest reconfiguration sequence, 
and showed that it is \NP-hard on bipartite graphs.  
They asked whether the problem is in \NP\ and whether the problem remains hard without any length restrictions.  
We settle the complexity of $\mc{M}$-ISR (and the shortest path variant) on bipartite graphs under all three models. 
We show that under the token jumping or token addition/removal model the problem is \NP-complete. 
For the token sliding model, we show that the problem remains \PSPACE-complete. 
Our \NP-completeness result comes as somewhat of a surprise, as reconfiguration problems are typically in \P\ or \PSPACE-complete~\cite{H13}. 
To the best of our knowledge, TAR-ISR on bipartite graphs is the first ``natural'' \NP-complete reconfiguration problem 
(that asks for the existence of a reconfiguration sequence of \emph{any} length). 

It is known~\cite{KMM12} that the token addition/removal model generalizes the token jumping model in the following sense.  
There exists a sequence between two independent sets $I$ and $J$, with $|I| = |J|$, under the TJ model if and only if there exists a  
sequence between them under the TAR model, with $k = |I| - 1$. Hence, we only consider the TAR model.  
In addition, TAR-ISR is easily seen to be equivalent to the following problem, namely {\sc Vertex Cover Reconfiguration} (VCR). 
We are given an $n$-vertex graph $G$, an integer $k$, and two vertex covers of $G$, $S$ and $T$, of size at most $k$. 
The goal is to determine whether there exists a sequence $\sigma = \langle Q_0, \ldots, Q_t \rangle$ satisfying the following. 
\begin{itemize}
\item $Q_0 = S$ and $Q_t = T$; 
\item $Q_i$ is a vertex cover of $G$ and $|Q_i| \leq k$, for $0 \leq i \leq t$;
\item $|Q_i \Delta Q_{i+1}| = 1$, where $0 \leq i < t$ and $Q_i \Delta Q_{i+1} = (Q_i \setminus Q_{i+1}) \cup (Q_{i+1} \setminus Q_{i})$ 
denotes the \emph{symetric difference} of $Q_i$ and $Q_{i+1}$. 
\end{itemize}
An alternative perspective on the VCR problem (and reconfiguration problems in general) is via the notion of the \emph{reconfiguration graph} $\mathcal{R}_k(G)$. 
Nodes in $\mathcal{R}_k(G)$ represent vertex covers of $G$ of size at most $k$ and two nodes $Q$ and $Q'$ 
are connected by an edge whenever $|Q \Delta Q'| = 1$. In other words, $Q'$ can be obtained from $Q$ by the addition or removal of a single vertex. 
An edge in $\mathcal{R}_k(G)$ is sometimes referred to as a reconfiguration step and a walk or path in 
this graph is a reconfiguration sequence. An equivalent formulation of VCR is then to 
determine whether $S$ and $T$ belong to the same connected component of $\mathcal{R}_k(G)$. 

\begin{figure}
\centering
    \begin{tikzpicture}[scale=.75, auto=left, remember picture ,every node/.style={circle},inner/.style={circle},outer/.style={circle}]

	\node at (1,2.75) {\large $S$ (source)};
	\node at (14,0.5) {\large $T$ (target)};
	
    \node[outer] (B1) at (1,5) {
    \begin{tikzpicture}
        \node [inner, circle, fill=gray!50, draw=black, inner sep=1pt] (a1)  {~3};
        \node [inner, circle, fill=gray!50, draw=black, inner sep=1pt, above=0.2cm of a1] (a2) {10};
        \node [inner, circle, fill=gray!50, draw=black, inner sep=1pt, above=0.2cm of a2] (a3) {10};
        \node [inner, circle, fill=white, draw=black, inner sep=1pt, above=0.2cm of a3] (a4) {~8};
        \node [inner, circle, fill=gray!50, draw=black, inner sep=1pt, right=0.2cm of a4] (a5) {14};
		\node [inner, circle, fill=white, draw=black, inner sep=1pt, below=0.2cm of a5] (a6) {~5};
		\node [inner, circle, fill=white, draw=black, inner sep=1pt, below=0.2cm of a6] (a7) {~2};
		\node [inner, circle, fill=white, draw=black, inner sep=1pt, below=0.2cm of a7] (a8) {10};
        \foreach \from/\to in {a4/a5,a3/a5,a3/a6,a3/a7,a2/a7,a2/a8,a1/a7,a1/a8}
        \draw (\from) -- (\to);
    \end{tikzpicture}};

    \node[outer] (B2) at (4,5) {
    \begin{tikzpicture}
        \node [inner, circle, fill=gray!50, draw=black, inner sep=1pt] (a1)  {~3};
        \node [inner, circle, fill=gray!50, draw=black, inner sep=1pt, above=0.2cm of a1] (a2) {10};
        \node [inner, circle, fill=gray!50, draw=black, inner sep=1pt, above=0.2cm of a2] (a3) {10};
        \node [inner, circle, fill=white, draw=black, inner sep=1pt, above=0.2cm of a3] (a4) {~8};
        \node [inner, circle, fill=gray!50, draw=black, inner sep=1pt, right=0.2cm of a4] (a5) {14};
		\node [inner, circle, fill=gray!50, draw=black, inner sep=1pt, below=0.2cm of a5] (a6) {~5};
		\node [inner, circle, fill=gray!50, draw=black, inner sep=1pt, below=0.2cm of a6] (a7) {~2};
		\node [inner, circle, fill=white, draw=black, inner sep=1pt, below=0.2cm of a7] (a8) {10};
        \foreach \from/\to in {a4/a5,a3/a5,a3/a6,a3/a7,a2/a7,a2/a8,a1/a7,a1/a8}
        \draw (\from) -- (\to);
    \end{tikzpicture}};

    \node[outer] (B3) at (7,5) {
    \begin{tikzpicture}
        \node [inner, circle, fill=gray!50, draw=black, inner sep=1pt] (a1)  {~3};
        \node [inner, circle, fill=gray!50, draw=black, inner sep=1pt, above=0.2cm of a1] (a2) {10};
        \node [inner, circle, fill=white, draw=black, inner sep=1pt, above=0.2cm of a2] (a3) {10};
        \node [inner, circle, fill=white, draw=black, inner sep=1pt, above=0.2cm of a3] (a4) {~8};
        \node [inner, circle, fill=gray!50, draw=black, inner sep=1pt, right=0.2cm of a4] (a5) {14};
		\node [inner, circle, fill=gray!50, draw=black, inner sep=1pt, below=0.2cm of a5] (a6) {~5};
		\node [inner, circle, fill=gray!50, draw=black, inner sep=1pt, below=0.2cm of a6] (a7) {~2};
		\node [inner, circle, fill=white, draw=black, inner sep=1pt, below=0.2cm of a7] (a8) {10};
        \foreach \from/\to in {a4/a5,a3/a5,a3/a6,a3/a7,a2/a7,a2/a8,a1/a7,a1/a8}
        \draw (\from) -- (\to);
    \end{tikzpicture}};
	
    \node[outer] (B4) at (10,5) {
    \begin{tikzpicture}
        \node [inner, circle, fill=gray!50, draw=black, inner sep=1pt] (a1)  {~3};
        \node [inner, circle, fill=gray!50, draw=black, inner sep=1pt, above=0.2cm of a1] (a2) {10};
        \node [inner, circle, fill=white, draw=black, inner sep=1pt, above=0.2cm of a2] (a3) {10};
        \node [inner, circle, fill=white, draw=black, inner sep=1pt, above=0.2cm of a3] (a4) {~8};
        \node [inner, circle, fill=gray!50, draw=black, inner sep=1pt, right=0.2cm of a4] (a5) {14};
		\node [inner, circle, fill=gray!50, draw=black, inner sep=1pt, below=0.2cm of a5] (a6) {~5};
		\node [inner, circle, fill=gray!50, draw=black, inner sep=1pt, below=0.2cm of a6] (a7) {~2};
		\node [inner, circle, fill=gray!50, draw=black, inner sep=1pt, below=0.2cm of a7] (a8) {10};
        \foreach \from/\to in {a4/a5,a3/a5,a3/a6,a3/a7,a2/a7,a2/a8,a1/a7,a1/a8}
        \draw (\from) -- (\to);
    \end{tikzpicture}};
	
    \node[outer] (B5) at (13,5) {
    \begin{tikzpicture}
        \node [inner, circle, fill=white, draw=black, inner sep=1pt] (a1)  {~3};
        \node [inner, circle, fill=white, draw=black, inner sep=1pt, above=0.2cm of a1] (a2) {10};
        \node [inner, circle, fill=white, draw=black, inner sep=1pt, above=0.2cm of a2] (a3) {10};
        \node [inner, circle, fill=white, draw=black, inner sep=1pt, above=0.2cm of a3] (a4) {~8};
        \node [inner, circle, fill=gray!50, draw=black, inner sep=1pt, right=0.2cm of a4] (a5) {14};
		\node [inner, circle, fill=gray!50, draw=black, inner sep=1pt, below=0.2cm of a5] (a6) {~5};
		\node [inner, circle, fill=gray!50, draw=black, inner sep=1pt, below=0.2cm of a6] (a7) {~2};
		\node [inner, circle, fill=gray!50, draw=black, inner sep=1pt, below=0.2cm of a7] (a8) {10};
        \foreach \from/\to in {a4/a5,a3/a5,a3/a6,a3/a7,a2/a7,a2/a8,a1/a7,a1/a8}
        \draw (\from) -- (\to);
    \end{tikzpicture}};
	
    \node[outer] (B6) at (16,5) {
    \begin{tikzpicture}
        \node [inner, circle, fill=white, draw=black, inner sep=1pt] (a1)  {~3};
        \node [inner, circle, fill=white, draw=black, inner sep=1pt, above=0.2cm of a1] (a2) {10};
        \node [inner, circle, fill=gray!50, draw=black, inner sep=1pt, above=0.2cm of a2] (a3) {10};
        \node [inner, circle, fill=white, draw=black, inner sep=1pt, above=0.2cm of a3] (a4) {~8};
        \node [inner, circle, fill=gray!50, draw=black, inner sep=1pt, right=0.2cm of a4] (a5) {14};
		\node [inner, circle, fill=gray!50, draw=black, inner sep=1pt, below=0.2cm of a5] (a6) {~5};
		\node [inner, circle, fill=gray!50, draw=black, inner sep=1pt, below=0.2cm of a6] (a7) {~2};
		\node [inner, circle, fill=gray!50, draw=black, inner sep=1pt, below=0.2cm of a7] (a8) {10};
        \foreach \from/\to in {a4/a5,a3/a5,a3/a6,a3/a7,a2/a7,a2/a8,a1/a7,a1/a8}
        \draw (\from) -- (\to);
    \end{tikzpicture}};
	
    \node[outer] (B7) at (6,0.5) {
    \begin{tikzpicture}
        \node [inner, circle, fill=white, draw=black, inner sep=1pt] (a1)  {~3};
        \node [inner, circle, fill=white, draw=black, inner sep=1pt, above=0.2cm of a1] (a2) {10};
        \node [inner, circle, fill=gray!50, draw=black, inner sep=1pt, above=0.2cm of a2] (a3) {10};
        \node [inner, circle, fill=white, draw=black, inner sep=1pt, above=0.2cm of a3] (a4) {~8};
        \node [inner, circle, fill=gray!50, draw=black, inner sep=1pt, right=0.2cm of a4] (a5) {14};
		\node [inner, circle, fill=white, draw=black, inner sep=1pt, below=0.2cm of a5] (a6) {~5};
		\node [inner, circle, fill=gray!50, draw=black, inner sep=1pt, below=0.2cm of a6] (a7) {~2};
		\node [inner, circle, fill=gray!50, draw=black, inner sep=1pt, below=0.2cm of a7] (a8) {10};
        \foreach \from/\to in {a4/a5,a3/a5,a3/a6,a3/a7,a2/a7,a2/a8,a1/a7,a1/a8}
        \draw (\from) -- (\to);
    \end{tikzpicture}};
	
    \node[outer] (B8) at (9,0.5) {
    \begin{tikzpicture}
        \node [inner, circle, fill=white, draw=black, inner sep=1pt] (a1)  {~3};
        \node [inner, circle, fill=white, draw=black, inner sep=1pt, above=0.2cm of a1] (a2) {10};
        \node [inner, circle, fill=gray!50, draw=black, inner sep=1pt, above=0.2cm of a2] (a3) {10};
        \node [inner, circle, fill=gray!50, draw=black, inner sep=1pt, above=0.2cm of a3] (a4) {~8};
        \node [inner, circle, fill=gray!50, draw=black, inner sep=1pt, right=0.2cm of a4] (a5) {14};
		\node [inner, circle, fill=white, draw=black, inner sep=1pt, below=0.2cm of a5] (a6) {~5};
		\node [inner, circle, fill=gray!50, draw=black, inner sep=1pt, below=0.2cm of a6] (a7) {~2};
		\node [inner, circle, fill=gray!50, draw=black, inner sep=1pt, below=0.2cm of a7] (a8) {10};
        \foreach \from/\to in {a4/a5,a3/a5,a3/a6,a3/a7,a2/a7,a2/a8,a1/a7,a1/a8}
        \draw (\from) -- (\to);
    \end{tikzpicture}};
	
    \node[outer] (B9) at (12,0.5) {
    \begin{tikzpicture}
        \node [inner, circle, fill=white, draw=black, inner sep=1pt] (a1)  {~3};
        \node [inner, circle, fill=white, draw=black, inner sep=1pt, above=0.2cm of a1] (a2) {10};
        \node [inner, circle, fill=gray!50, draw=black, inner sep=1pt, above=0.2cm of a2] (a3) {10};
        \node [inner, circle, fill=gray!50, draw=black, inner sep=1pt, above=0.2cm of a3] (a4) {~8};
        \node [inner, circle, fill=white, draw=black, inner sep=1pt, right=0.2cm of a4] (a5) {14};
		\node [inner, circle, fill=white, draw=black, inner sep=1pt, below=0.2cm of a5] (a6) {~5};
		\node [inner, circle, fill=gray!50, draw=black, inner sep=1pt, below=0.2cm of a6] (a7) {~2};
		\node [inner, circle, fill=gray!50, draw=black, inner sep=1pt, below=0.2cm of a7] (a8) {10};
        \foreach \from/\to in {a4/a5,a3/a5,a3/a6,a3/a7,a2/a7,a2/a8,a1/a7,a1/a8}
        \draw (\from) -- (\to);
    \end{tikzpicture}};

	\draw[thick,->] (2.25,5) -- (2.75,5);
	\draw[thick,->] (5.25,5) -- (5.75,5);
	\draw[thick,->] (8.25,5) -- (8.75,5);
	\draw[thick,->] (11.25,5) -- (11.75,5);
	\draw[thick,->] (14.25,5) -- (14.75,5);
	\draw[thick,->] (17.25,5) -- (17.75,5);

	\draw[thick,->] (3.5,0.5) -- (4.5,0.5);
	\draw[thick,->] (7.25,0.5) -- (7.75,0.5);
	\draw[thick,->] (10.25,0.5) -- (10.75,0.5);
	
  \end{tikzpicture}
\caption{Example of a non-monotone sequence going to a local minima ($k = 44$). Each edge $xy$ corresponds to a biclique with $x$ vertices on one side 
and $y$ vertices on the other. Vertices colored gray belong to the vertex cover. Note that vertices that are added or removed more than once 
do not belong to $S \Delta T$.}
\label{fig-example}
\end{figure}
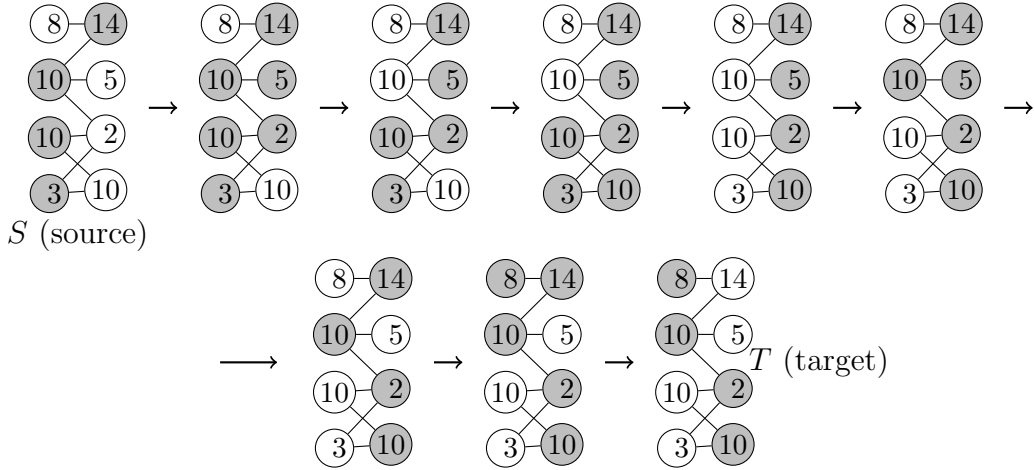

To prove the \NP-hardness of VCR (Section~\ref{sec-hardness}), we consider instances $(G,S,T,k)$ of the problem where $G$ is an $n$-vertex 
bipartite graph with bipartition $(L,R)$, $S = L$, and $T = R$. Informally, we call such instances the ``left-to-right instances''. 
We show that there is a reconfiguration sequence from $L$ to $R$ in $\mc{R}_k(G)$ if and only if 
the treewidth of the cobipartite graph $\overline{G}$ is at most $k$, where $\overline{G}$ is 
obtained from $G$ by adding all edges between vertices in $L$ and adding all edges between vertices in $R$.   
We obtain the aforementioned equivalence by relating left-to-right instances of VCR to the cops-and-robber game played on $\overline{G}$. 
The seminal result of Seymour and Thomas~\cite{SeymourT93} establishes the equivalence 
between the cops-and-robber game and computing the treewidth of the underlying graph $\overline{G}$.  
That is, the number of cops needed to catch a robber in $\overline{G}$ is exactly equal to the treewidth of $\overline{G}$ plus one. 
Computing treewidth is known to be \NP-complete, even for cobipartite graphs~\cite{Arnborg:1987:CFE:37170.37183,DBLP:journals/dam/BodlaenderT97,doi:10.1137/0602010}. 
This implies that VCR is \NP-complete when restricted to left-to-right instances, or more generally to instances where $S \cap T = \emptyset$. 
To show membership in \NP\ for instances that are not necessarily left-to-right (Section~\ref{sec-membership}), we prove that the diameter of 
every connected component of $\mathcal{R}_k(G)$, for $G$ bipartite, is at most $\Oh(n^4)$.  
While we believe that $\Oh(n^4)$ is an overestimation, we know that there are instances of VCR where vertices have to be ``touched'' (added or removed) more than once. 
An example of such an instance (where vertices need to be added and removed twice) is shown in Figure~\ref{fig-example}. 

Finally, we consider ISR under the token sliding model in Section~\ref{sec-pspace}. This problem is commonly known in the literature 
as the {\sc Token Sliding} problem. We prove \PSPACE-completeness of {\sc Token Sliding} in bipartite graphs by a reduction from 
a variant of the {\sc Word Reconfiguration} problem, first introduced by Wrochna~\cite{DBLP:journals/corr/Wrochna14}.

\section{Preliminaries}\label{sec-prelim}
We denote the set of natural numbers by $\mathbb{N}$. For $n \in \mathbb{N}$, we let $[n] = \{1, 2, \cdots, n\}$. 
We assume that each graph $G$ is finite, simple, and undirected. We let $V(G)$ and $E(G)$ denote the vertex set and edge set of $G$, respectively. 
We use $\allsets$ to denote the set of all subsets of $V(G)$ of cardinality at most $k$, where $k$ is a non-zero positive integer. 
The {\em open neighborhood} of a vertex $v$ is denoted by $N_G(v) = \{u \mid uv \in E(G)\}$ and the
{\em closed neighborhood} by $N_G[v] = N_G(v) \cup \{v\}$.
For a set of vertices $Q \subseteq V(G)$, we define $N_G(Q) = \{v \not\in Q \mid uv \in E(G), u \in Q\}$ and $N_G[Q] = N_G(Q) \cup Q$.
The subgraph of $G$ induced by $Q$ is denoted by $G[Q]$, where $G[Q]$ has vertex set 
$Q$ and edge set $\{uv \in E(G) \mid u,v \in Q\}$. We let $G - Q = G[V(G) \setminus Q]$. 
For a pair of vertices $u$ and $v$ in $V(G)$, by $\textsf{dist}_G(u,v)$ we denote the length of a shortest path
between $u$ and $v$ in $G$ (measured in number of edges and set to $\infty$ if $u$ and $v$ belong to different connected components). 
Given a graph $G$ and a set $Q \subseteq V(G)$, by \emph{cliquifying $Q$} we denote the operation that adds all 
missing edges between vertices in $Q$, resulting in a new graph $G'$. Given $G$ and a vertex $u \in V(G)$, 
by \emph{duplicating $u$} we denote the operation that adds a new vertex $v$ connected to all vertices in $N_G(u)$, resulting in a new graph $G'$. 
In $G'$, $u$ and $v$ are \emph{twins}.  
A graph $G$ is {\em bipartite} if the vertex set of $G$ can be partitioned into two disjoint sets $L$ and $R$, i.e. $V(G) = L \cup R$, 
where $G[L]$ and $G[R]$ are edgeless. A graph $G$ is {\em cobipartite} if the vertex set of $G$ can be partitioned 
into two disjoint sets $L$ and $R$, where $G[L]$ and $G[R]$ are cliques.  

\paragraph{Canonical path and tree decompositions.}
A \emph{tree decomposition}~\cite{DBLP:journals/actaC/Bodlaender93,DBLP:journals/jct/RobertsonS84} of a graph $G$ is a pair $\mathcal{T} = (T, \{B_t\}_{t \in V(T)})$,  
where $T$ is a tree whose every vertex $t$ is assigned a vertex
subset $B_t \subseteq V(G)$, called a \emph{bag}, such that the following three conditions hold.
\begin{itemize}
\item (P1) $\bigcup_{t \in V(T)}{B_t} = V(G)$, in other words, every vertex of $G$ is in at least one bag;
\item (P2) For every $uv \in E(G)$, there exists $t \in V(T)$ such that bag $B_t$ contains $u$ and $v$;
\item (P3) For every $u \in V(G)$, the set $T_u = \{t \in V(T) \mid u \in B_t\}$, i.e., the set
of vertices whose corresponding bags contain $u$, induces a connected subtree of $T$. 
\end{itemize}
The \emph{width} of tree decomposition $\mathcal{T}= (T, \{B_t\}_{t \in V(T)})$ equals $\textsf{max}_{t \in V(T)}\{|B_t| - 1\}$, that is, the maximum size of any bag minus $1$.  
The \emph{treewidth} of a graph $G$, denoted by $\textsf{tw}(G)$, is the minimum possible width of a tree decomposition of $G$. 
It is convenient to think of tree decompositions
as rooted trees. That is, for a tree decomposition $\mathcal{T}$ we distinguish
one vertex $r$ of $T$ which will be the root of $T$. This introduces natural
parent-child and ancestor-descendant relations in the tree $T$. We will say
that such a rooted tree decomposition is \emph{nice} if the following
conditions are satisfied. $B_r = \emptyset$ and $B_\ell = \emptyset$ for every leaf $\ell$ of $T$. 
In other words, all the leaves as well as the root contain empty bags.
Every non-leaf vertex of $T$ is of one of the following three types.
\begin{itemize}
\item \emph{Introduce vertex}: a vertex $t$ with exactly one child $t'$ such that $B_t = B_{t'} \cup \{v\}$, for some 
vertex $v \not\in B_{t'}$; we say that $v$ is \emph{introduced} at $t$;
\item \emph{Forget vertex}: a vertex $t$ with exactly one child $t'$ such that 
$B_t = B_{t'} \setminus \{w\}$ for some vertex $w \in B_{t'}$; we say that $w$ is \emph{forgotten} at $t$;
\item \emph{Join vertex}: a vertex $t$ with two children $t_1$ and $t_2$ such that $B_t = B_{t_1} = B_{t_2}$. 
\end{itemize}

A \emph{(nice) path decomposition} of a graph $G$ is simply a (nice) tree decomposition where $T$ must be a path. 
It will be convenient to denote a (nice) path decomposition by a sequence $\mathcal{P} = \{B_1, B_2, \ldots, B_p\}$.  
The \emph{pathwidth} of a graph $G$, denoted by $\textsf{pw}(G)$, is the minimum possible width of a path decomposition of $G$. 
Note that for a nice path decomposition we only have introduce and forget vertices. 

\begin{proposition}\label{pr-clique}
For every clique $C$ of a graph $G$ and any (nice) path or tree decomposition of $G$, there exists a bag $B$ in the (nice) path or 
tree decomposition of $G$ such that $V(C) \subseteq B$. 
\end{proposition}

\section{NP-completeness under the token addition/removal model}\label{sec-np-comp}
We denote an instance of the {\sc Vertex Cover Reconfiguration} problem by $(G, S, T, k)$,
where $G$ is the input graph, $S$ and $T$ are the {\em source} and {\em target} 
vertex covers, respectively, and $k$ is the {\em maximum allowed capacity}. 
The reconfiguration graph $\mc{R}_k(G)$ contains a node for each vertex cover $Q$ of $G$ of size at most $k$. 
Two nodes $Q$ and $Q'$ are adjacent in $\mc{R}_k(G)$ whenever $|Q \Delta Q'| = 1$.  
To avoid confusion, we refer to nodes in reconfiguration graphs, as distinguished from vertices in the input graph. 
By a slight abuse of notation, we use upper case letters to refer to both a node in the reconfiguration
graph as well as the corresponding vertex cover. 
For any node $Q \in V(\mc{R}_k(G))$, the quantity $k - |Q|$ corresponds to the {\em available capacity} at $Q$. 
Given $Q \in V(\mc{R}_k(G))$, let $\mc{C}$ denote the connected component of $\mc{R}_k(G)$ containing $Q$. 
We say $Q$ is a {\em local minima} if there exists no $Q' \in V(\mc{C})$ such that $|Q'| < |Q|$. 

We use two representations of reconfiguration sequences. 
The first representation consists of a sequence of vertex covers, 
$\sigma = \langle Q_0, Q_1, \ldots, Q_{t - 1}, Q_t \rangle$. Given $\sigma$, we associate it with a sequence of {\em edit operations} as follows. 
We assume all vertices of $G$ are labeled from $1$ to $n$, i.e., $V(G) = \{v_1, v_2, \ldots, v_n\}$. 
We let $\mc{M}^+ = \{v^+_1, \ldots, v^+_n\}$ and $\mc{M}^- = \{v^-_1, \ldots, v^-_n\}$ denote the
sets of {\em addition markers} and {\em removal markers}, respectively. 
An {\em edit sequence} $\eta = \langle m_1, m_2, \ldots, m_{t-1}, m_t \rangle$ is an ordered sequence of
elements obtained from the full set of {\em markers} $\mc{M}^+ \cup \mc{M}^-$, i.e., $m_i \in \mc{M}^+ \cup \mc{M}^-$. 
Here $v^+_i$ stands for ``add vertex $v_i$'' and $v^-_j$ stands for ``remove vertex $v_j$'', $1 \leq i,j \leq n$.  
We say a vertex $v \in V(G)$ is {\em touched} in the course of a reconfiguration
sequence if $v$ is either added ($v^+ \in \eta$) or removed ($v^- \in \eta$) at least once. 
We say a reconfiguration sequence is \emph{monotone} if it touches no vertex more than once. 
Given an edit sequence $\eta$ and a vertex cover $S$, We say $\eta$ is {\em valid at $S$} if applying $\eta$ starting from $S$ results 
in a reconfiguration sequence, i.e., the sequence corresponds to a walk in the reconfiguration graph $\mc{R}_k(G)$. 
In other words, $\eta = \langle m_1, m_2, \ldots, m_{t-1}, m_t \rangle$ is valid at $S$ if and only if there exists a sequence 
$\sigma = \langle Q_0 = S, Q_1, \ldots, Q_{t - 1}, Q_t \rangle$ such that 
$Q_i$ is a vertex cover of $G$ of size at most $k$, for $0 \leq i \leq t$, and 
if $m_{i+1} \in \{v^+_j,v^-_j\}$ then $Q_i \Delta Q_{i+1} = \{v_j\}$, for $0 \leq i < t$. 
We let $\eta(S)$ denote the vertex cover obtained after applying $\eta$ starting at $S$. Given $\eta$ and $\eta'$, we 
use $\eta \cdot \eta'$ to denote the concatenation of both sequences. 

\subsection{NP-hardness}\label{sec-hardness}
To prove \NP-hardness we will show an equivalence between VCR and the cops-and-robber game. 
Let us start by formally describing the game. 
The game is played on a finite, undirected, and connected graph $G$. Throughout the game, a robber 
is standing on some vertex in $V(G)$. The robber can, at any time, run 
at ``infinite'' speed to any other vertex along a path of the graph. 
However, running through a cop is not permitted. There are $k$ cops, each of them is 
either standing on a vertex or in a helicopter (temporarily removed from the graph). The cops can see 
the robber at all times. 
The objective of the cops is to land a helicopter on the vertex occupied by the robber, and 
the robber's objective is to elude capture. Note that, since they are equipped with helicopters, 
cops are not constrained to moving along paths of the graph. The robber can see the helicopter approaching 
the landing spot and can run to a new vertex before the helicopter actually lands (when possible). 
Therefore, the only way for the cops to capture the robber is by having a cop land on the vertex $v$ occupied by the robber 
while all neighbors of $v$ are also occupied by cops. 

A \emph{state} in the game is a pair $(X,F)$, where $X \in \allsets$ and $F$ is an \emph{$X$-flap}, i.e., the vertex set of a component of $G - X$. 
$X$ is the set currently occupied by cops and $F$ tells us where the robber is; since he can run 
arbitrarily fast, the only information we need is which component of $G - X$ contains him. 
The initial state is $(X_0,F_0)$, where $X_0 = \emptyset$ and $F_0$ is the flap chosen by the robber. 
At round $i \geq 1$ of the game, we have $(X_{i-1},F_{i-1})$ and the cops pick a new set $X_i \subseteq \allsets$ 
such that $|X_i \Delta X_{i-1}| = 1$; either a helicopter lands on a vertex or a cop leaves the graph on a helicopter. 
Then, the robber chooses (if possible) an $X_i$-flap $F_i$ such that $F_i \subseteq F_{i-1}$ or 
$F_{i-1} \subseteq F_i$. If $F_i \subseteq X_i$ then the cops win. Otherwise, the game continues with round $i + 1$. 
If there is a winning strategy for the cops, we say that $\leq k$ cops can search the graph. 
If in addition the cops can always win in such a way that the sequence $\langle X_0$, $X_1$, $\ldots \rangle$ satisfies 
$X_i \cap X_{i''} \subseteq X_{i'}$, for $i \leq i' \leq i''$, then we say that $\leq k$ cops can monotonely search the graph. 
In other words, a monotone strategy implies that cops never return to a vertex that has been previously vacated. 
The following theorem is due to Seymour and Thomas~\cite{SeymourT93}. 

\begin{theorem}\label{th-seymour}
Let $G$ be a graph and $k$ be a non-zero positive integer. Then the following are equivalent:
\begin{itemize}
\item $\leq k$ cops can search $G$; 
\item $\leq k$ cops can monotonely search $G$; 
\item $G$ has treewidth at most $k - 1$. 
\end{itemize}
\end{theorem}

It is well-known that computing the pathwidth or the treewidth of a graph is an \NP-hard problem~\cite{Arnborg:1987:CFE:37170.37183}. 
In what follows, we let $(G,S = L,T = R,k)$ be an instance of VCR, where $G$ is bipartite, 
$V(G) = L \cup R$, $S = L$, and $T = R$ (a left-to-right instance). 
We let $\overline{G}$ be the cobipartite graph obtained from $G$ by first cliquifying $L$ and then cliquifying $R$. 
We let $V(\overline{G}) = \overline{L} \cup \overline{R}$. 
The treewidth of a cobipartite graph is known to be equal to its pathwidth~\cite{Habib1994}.  
Moreover, it is \NP-hard to determine, given a cobipartite graph $\overline{G}$ and an integer $k$, whether 
$\overline{G}$ has treewidth at most $k$~\cite{Arnborg:1987:CFE:37170.37183,DBLP:journals/dam/BodlaenderT97,doi:10.1137/0602010}. 
Lemmas~\ref{le-hardness1} and~\ref{le-hardness2} below, combined with Theorem~\ref{th-seymour} and the \NP-hardness of 
computing treewidth of cobipartite graphs, imply the \NP-hardness of VCR.

\begin{lemma}\label{le-hardness1}
If $L$ and $R$ belong to the same connected component of $\mc{R}_k(G)$ then $k$ cops can search the cobipartite graph $\overline{G}$. 
\end{lemma}

\begin{proof}
Let $\langle Q_0 = L, Q_1, \ldots, Q_{t-1}, Q_t = R \rangle$ denote a reconfiguration sequence from $L$ to $R$. 
Let $L = \overline{L} = \{u_1, \ldots, u_\ell\}$ and $R = \overline{R} = \{v_1, \ldots, v_r\}$.
We claim that $\langle X_0 = \{u_1\}$, $X_1 = X_0 \cup \{u_2\}$, $\ldots$, $X_\ell = X_{\ell - 1} \cup \{u_\ell\} = \overline{L}$, 
$X_{\ell + 1} = Q_1$, $\ldots$, $X_{\ell + t} = Q_t = \overline{R} \rangle$ 
is a winning strategy for the cops (on $\overline{G}$). Since all vertex covers (of $G$) in the reconfiguration sequence 
have size at most $k$, it follows that $|X_i| \leq k$, $0 \leq i \leq \ell + t$. 
Note that at round $i = \ell$ the robber must move to some vertex in $\overline{R}$ (as all vertices in $\overline{L}$ are occupied by cops). 
Moreover, at any round $i > \ell$ the robber cannot move to a vertex in $\overline{L}$. Assume otherwise. Then, 
there exists some path starting at the robber's vertex $v^\star \in \overline{R}$ and ending at some vertex $u^\star \in \overline{L}$. Let 
$vu$ be the first edge on this path that takes the robber from $\overline{R}$ to $\overline{L}$. Both $u$ and $v$ are not occupied by cops, contradicting 
the fact that $Q_i$ is vertex cover of $G$. Hence, the robber is captured at round $\ell + t$, i.e., when all of $\overline{R}$ is occupied by cops. 
\end{proof}

\begin{lemma}\label{le-hardness2}
If $k$ cops can search the cobipartite graph $\overline{G}$ then $L$ and $R$ belong to the same connected component of $\mc{R}_k(G)$. 
Moreover, there exists a monotone reconfiguration sequence from $L$ to $R$ in $\mc{R}_k(G)$. 
\end{lemma}

\begin{proof}
As there exists a winning strategy for the cops, by Theorem~\ref{th-seymour}, we can assume that $\overline{G}$ has treewidth at most $k - 1$. 
Since the treewidth of a cobipartite graph is equal to its pathwidth~\cite{Habib1994}, let $\mathcal{P} = \{B_1, B_2, \ldots, B_p\}$ 
be a nice path decomposition of $\overline{G}$ of width at most $k - 1$.  
Combining Proposition~\ref{pr-clique} with the fact that $\overline{G}[\overline{L}]$ 
and $\overline{G}[\overline{R}]$ are cliques, we know 
that there exists $q$ and $q'$ such that $1 \leq q < q' \leq p$ and either $\overline{L} \subseteq B_q$ and $\overline{R} \subseteq B_{q'}$ 
or $\overline{R} \subseteq B_q$ and $\overline{L} \subseteq B_{q'}$ 
(if $q = q'$ then $k \geq |\overline{L}| + |\overline{R}| = |V(\overline{G})| = |V(G)|$ and the lemma trivially holds).  
Assume, without loss of generality, that $\overline{L} \subseteq B_q$ and $\overline{R} \subseteq B_{q'}$ 
(if the assumption does not hold then we consider the reverse path decomposition). 
Let $q^\star = q' - q$, $Z = B_q \setminus \overline{L} = \{z_1, \ldots, z_{|Z|}\}$, and $Y = B_{q'} \setminus \overline{R} = \{y_1, \ldots, y_{|Y|}\}$. 
Consider the following three sequences:  
\begin{itemize}
\item $\sigma_1 = \langle Q_0, Q_1, \ldots, Q_{|Z| - 1}, Q_{|Z|} \rangle$, where $Q_0 = \overline{L}$ and $Q_i = Q_{i - 1} \cup \{z_i\}$, $1 \leq i \leq |Z|$ 
(note that $Q_{|Z|} = B_q$);
\item $\sigma_2 = \langle Q_{|Z| + 1}, Q_{|Z|+2}, \ldots, Q_{|Z| + q^\star - 1}, Q_{|Z| + q^\star} \rangle$, where $Q_{|Z| + i} = B_{q + i}$, $1 \leq i \leq q^\star$ 
(note that $Q_{|Z| + q^\star} = B_{q + q\star} = B_{q'}$);
\item $\sigma_3 = \langle Q_{|Z| + q^\star + 1}, Q_{|Z| + q^\star + 2}, \ldots, Q_{|Z| + q^\star + |Y| - 1}, Q_{|Z| + q^\star + |Y|} \rangle$, 
where $Q_{|Z| + q^\star + 1} = Q_{|Z| + q^\star} \setminus \{y_1\}$ and $Q_{|Z| + q^\star + i} = Q_{|Z| + q^\star + i - 1} \setminus \{y_i\}$, $1 < i \leq |Y|$ 
(note that $S_{|Z| + q^\star + |Y|} = S_{|Z| + q^\star} \setminus Y = B_{q'} \setminus Y = \overline{R}$).
\end{itemize} 
We claim that $\sigma = \sigma_1 \cdot \sigma_2 \cdot \sigma_3$ is a monotone reconfiguration sequence from $L$ to $R$ in $\mc{R}_k(G)$. 
First, note that, by our construction of $\sigma_1$ and $\sigma_3$ and the fact that $\mathcal{P}$ is a nice path decomposition, we have 
$|Q \Delta Q'| = 1$ for any two consecutive sets in $\sigma$. 
Second, since the size of the bags is at most $k$, all the sets in $\sigma$ have size at most $k$; 
all sets appearing before $B_q = Q_{|Z|}$ and after $B_{q'} = Q_{|Z| + q^\star}$ have size strictly less than $|B_q|$ and $|B_{q'}|$, respectively. 
Moreover, as both $L$ and $R$ are vertex covers of $G$, it follows that $Q_i$ is also a vertex cover of $G$ for $i \leq |Z|$ ($L$ is a subset of every such $Q_i$) 
and $j \geq |Z| + q^\star$ ($R$ is a subset of every such $Q_j$).  
Hence, it remains to show that every set $Q_i$ in $\sigma_2$ is a vertex cover of $G$. Assume that there exists a set $Q_i$ in 
$\sigma_2$ that is not a vertex cover of $G$. Then there exists an uncovered edge $uv$ with $u \in L$ and $v \in R$. 
Since $L \subseteq Q_{|Z|}$ and $R \subseteq Q_{|Z| + q^\star}$, it must be the case that $u \not\in Q_i$ ($u$ was forgotten) and $v \not\in Q_i$ ($v$ was not yet introduced). 
This implies that there exists some other bag $B_j$, $j \neq i$, covering the edge $uv$ (by property (P2) of tree decompositions). 
But in both of the cases $j < i$ and $j > i$, we contradict property (P3) of tree decompositions; as the 
bags containing $u$ (and $v$) must appear consecutively in $\mc{P}$. Finally, the monotonicity of $\sigma$ follows from our construction 
(for $\sigma_1$ and $\sigma_3$) and property (P3) of tree decompositions (for $\sigma_2$).   
\end{proof}

\begin{theorem}\label{th-hardness}
Let $(G,S,T,k)$ be an instance of VCR where $S \cap T = \emptyset$. 
Let $\overline{G}$ be the graph obtained from $G$ by first cliquifying $S$ and then cliquifying $T$. 
Then the following are equivalent: 
\begin{itemize}
\item There exists a reconfiguration sequence from $S$ to $T$ in $\mc{R}_k(G)$;
\item There exists a monotone reconfiguration sequence from $S$ to $T$ in $\mc{R}_k(G)$; 
\item $\overline{G}$ has treewidth/pathwidth at most $k - 1$. 
\end{itemize}
\end{theorem}

\begin{proof}
Note that any graph $G$ having two disjoint vertex covers $S$ and $T$ must be bipartite; 
there can be no edges with both endpoints in $S$, no edges with both endpoints in $T$, and any vertex 
in $V(G) \setminus (S \cup T)$ must be isolated. We can safely delete all isolated vertices in $V(G) \setminus (S \cup T)$ to obtain a graph $G'$. 
In other words, $\mc{R}_k(G')$ is an induced subgraph of $\mc{R}_k(G)$. 
Hence, combining Lemmas~\ref{le-hardness1} and~\ref{le-hardness2} with Theorem~\ref{th-seymour}, we get the claimed equivalences.  
\end{proof}

Theorem~\ref{th-hardness} already implies that VCR on bipartite graphs is \NP-complete if we restrict ourselves to instances where $S \cap T = \emptyset$.

\subsection{Membership in NP}\label{sec-membership}
To prove membership in \NP\ (without any restrictions on $S \cap T$) we prove the following. 

\begin{theorem}\label{th-bounded-diameter}
Let $G$ be an $n$-vertex bipartite graph and let $k$ be a non-zero positive integer. Let $S$ and $T$ be two vertex covers 
of $G$ of size at most $k$. If $\textsf{dist}_{\mathcal{R}_k(G)}(S,T)$ is finite then 
$\textsf{dist}_{\mathcal{R}_k(G)}(S,T) = \Oh(n^4)$. 
\end{theorem}

\noindent
We proceed in two stages. In the first stage, we prove that finding a reconfiguration sequence 
from a vertex cover $S$ to any vertex cover of size at least one less can be done in a monotone way. Formally, we show the following.  

\begin{lemma}\label{le-onedown}
Given an $n$-vertex bipartite graph $G$, an integer $k$, and a vertex cover $S$ of $G$ of size at most $k$, then 
either $S$ is a local minima in $\mathcal{R}_k(G)$ or there exists a vertex cover $S'$ such that 
$|S'| < |S|$ and $\textsf{dist}_{\mathcal{R}_k(G)}(S,S') \leq n$. 
\end{lemma}
 
Assuming Lemma~\ref{le-onedown}, we know that given $S$ and $T$, we can in $\Oh(n^2)$ steps reconfigure $S$ and $T$ to local minimas $S'$ and $T'$, respectively. 
If $S$ and $T$ live in a connected component of $\mc{R}_k(G)$ with only one local minima then we are done. 
Unfortunatly this is not true in general. 
We remedy the situation by ``perturbing'' the vertex covers of $G$ to guarantee uniqueness in the following sense. 
Let $S^\star$ be a local minima in some connected component $\mc{C}$ of $\mc{R}_k(G)$. 
We let $H$ be the graph obtained from $G$ by duplicating each vertex $u \in S^\star$ a total of $2n - 1$ times (resulting in $2n$ copies of $u$) 
and duplicating each vertex $v \not\in S^\star$ a total of $2n$ times (resulting in $2n + 1$ copies of $v$). 
We set $\mu = 2nk + 2n - 1$. We call $H$ the \emph{perturbation} of $G$. 
For a vertex $v \in V(G)$, we define $f(v)$ to be the set of twins in $V(H)$ corresponding to $v$. 
By a slight abuse of notation, we generalize the function $f$ to sets as follows. 
For $X \subseteq V(G)$, we let $f(X) = \bigcup_{x \in X}f(x)$. Finally, for $\widehat{X} \subseteq V(H)$, we define the 
inverse of $f$ as $f^{-1}(\widehat{X}) = \{v \in V(G) \mid f(v) \subseteq \widehat{X}\}$. 
Intuitively, we will show that $\mc{R}_{\mu}(H)$ ``almost'' preserves the distances from 
$\mc{R}_{k}(G)$ and that the component of interest in $\mc{R}_{\mu}(H)$ has 
a unique local minima, i.e., the component containing $f(S^\star)$ will have $f(S^\star)$ as its unique local minima. 
In more technical terms, one can think of $\mc{R}_k(G)$ as having a low distortion embedding into $\mc{R}_{\mu}(H)$. 
The properties we require from the graph $H$ are captured in the following lemmas. 
Note that $|V(H)| \leq (2n + 1)n = \Oh(n^2)$ and $\mu = 2nk + 2n - 1 = \Oh(n^2)$. 

\begin{lemma}\label{le-vc-size}
If $Q$ is a vertex cover of $G$ then $f(Q)$ is a vertex cover of $H$ and $2n|Q| \leq |f(Q)| \leq (2n + 1)|Q|$. 
Moreover, if $\widehat{Q}$ is a vertex cover of $H$ then $f^{-1}(\widehat{Q})$ is a vertex cover of 
$G$ and $|f^{-1}(\widehat{Q})| \leq \lfloor \frac{|\widehat{Q}|}{2n} \rfloor$. 
\end{lemma}

\begin{proof}
Every edge $uv \in E(G)$ corresponds to a biclique in $H$, i.e., $H[f(u) \cup f(v)]$ is a biclique. 
By our construction of $H$, there is a surjective mapping from edges in $H$ to edges in $G$. 
That is, we can map every edge $yz \in E(H)$ to the unique edge $uv \in E(G)$ such that $yz \in E(H[f(u) \cup f(v)])$. 
Moreover, every edge of $G$ is mapped to (in the sense described above) by at least one edge of $H$. 
In other words, for all $uv \in E(G)$ there exists at least one edge $yz \in E(H)$ such that $yz \in E(H[f(u) \cup f(v)])$, 
and for all other edges $u'v' \in E(G)$ we have $yz \not\in E(H[f(u') \cup f(v')])$. 

If $Q$ is a vertex cover of $G$ then, for every edge $uv \in E(G)$, either $u \in Q$ or $v \in Q$. 
This implies that either $f(u) \subseteq f(Q)$ or $f(v) \subseteq f(Q)$. Therefore, given the surjective mapping 
from $E(H)$ to $E(G)$, $f(Q)$ must be a vertex cover of $H$. 
Since $|f(v)| = 2n$ if $v \in S^\star$ and $|f(v)| = 2n + 1$ otherwise, we have $2n|Q| \leq |f(Q)| \leq (2n + 1)|Q|$. 

Similarly, if $\widehat{Q}$ is a vertex cover of $H$, and since $H[f(u) \cup f(v)]$ is a biclique, then 
either $f(u) \subseteq \widehat{Q}$ or $f(v) \subseteq \widehat{Q}$ (for all $uv \in E(G)$). 
Therefore, either $u \in f^{-1}(\widehat{Q})$ or $v \in f^{-1}(\widehat{Q})$, as needed. 
To see why $|f^{-1}(\widehat{Q})| \leq \lfloor \frac{|\widehat{Q}|}{2n} \rfloor$, note that for every vertex cover 
$\widehat{Q}$ of $H$ we have $|\widehat{Q}| = (2n+1)|X| + 2n|Y| + |Z|$, where $X = \{v \in V(G) \setminus S^\star \mid f(v) \subseteq \widehat{Q}\}$, 
$Y = \{v \in S^\star \mid f(v) \subseteq \widehat{Q}\}$, and $Z = \widehat{Q} \setminus \bigcup_{v \in X \cup Y}f(v)$. 
Hence, $|\widehat{Q}| = 2n(|X| + |Y|) + |X| + |Z|$ and $|f^{-1}(\widehat{Q})| = |X| + |Y| \leq \lfloor \frac{|\widehat{Q}|}{2n} \rfloor$. 
\end{proof}

\begin{corollary}\label{co-vc-size}
Let $Q$ be a vertex cover of $G$ and let $\widehat{Q}$ be a vertex cover of $H$. 
\begin{itemize}
\item If $|Q| \leq k$ then $|f(Q)| \leq (2n + 1)k = 2nk + k < \mu$; 
\item If $|Q| \geq k + 1$ then $|f(Q)| \geq (k + 1)2n = 2nk + 2n > \mu$; 
\item If $|\widehat{Q}| \leq \mu$ then $|f^{-1}(\widehat{Q})| \leq \lfloor \frac{2nk + 2n - 1}{2n} \rfloor \leq k$; 
\end{itemize}
\end{corollary}

\begin{lemma}\label{le-inequality1}
$\forall S,T \in V(\mc{R}_k(G))$, if $\textsf{dist}_{\mathcal{R}_k(G)}(S,T)$ is finite then $\textsf{dist}_{\mathcal{R}_{\mu}(H)}(f(S),f(T))$ is finite 
and $\textsf{dist}_{\mathcal{R}_k(G)}(S,T) \leq \textsf{dist}_{\mathcal{R}_{\mu}(H)}(f(S),f(T)) \leq (2n + 1) \cdot \textsf{dist}_{\mathcal{R}_k(G)}(S,T)$. 
\end{lemma}

\begin{proof}
Let $\sigma = \langle Q_0 = S, Q_1 ,\ldots, Q_{t-1}, Q_t = T \rangle$ denote a shortest reconfiguration sequence from $S$ to $T$ in $\mathcal{R}_{k}(G)$. 
By Corollary~\ref{co-vc-size}, we know that $|f(Q_i)| \leq \mu$ and therefore $f(Q_i) \in V(\mathcal{R}_{\mu}(H))$, for all $i$. 
Hence, to obtain a corresponding reconfiguration sequence $\widehat{\sigma}$ in $\mathcal{R}_{\mu}(H)$ (from $f(S)$ to $f(T)$) it is 
enough to replace the addition or removal of any vertex $v$ by the addition or removal of $f(v)$ (one vertex at a time). 
It is not hard to see that every set in $\widehat{\sigma}$ is a vertex cover of $\widehat{G}$ of size at most $(2n + 1)k < \mu$ 
(and the size of the symmetric difference of consecutive sets is exactly one). 

Assume that there exists a reconfiguration sequence 
$\tilde{\sigma} = \langle \tilde{Q}_0 = f(S), \tilde{Q}_1 ,\ldots, \tilde{Q}_{\ell-1}, \tilde{Q}_\ell = f(T) \rangle$ 
from $f(S)$ to $f(T)$ in $\mathcal{R}_{\mu}(H)$ such that $|\tilde{\sigma}| < |\sigma|$. Consider the sequence $\sigma'$ obtained from 
$\tilde{\sigma}$ by replacing every set $\tilde{Q} \in \tilde{\sigma}$ by $f^{-1}(\tilde{Q})$ and then deleting duplicate sets. 
We claim that $\sigma'$ is a reconfiguration sequence from $S$ to $T$ in $\mathcal{R}_{k}(G)$. However, $|\sigma'| \leq \tilde{\sigma} < |\sigma|$, 
contradicting our assumption that $\sigma$ is a shortest reconfiguration sequence from $S$ to $T$ in $\mc{R}_k(G)$. 
Therefore, $\textsf{dist}_{\mathcal{R}_k(G)}(S,T) \leq \textsf{dist}_{\mathcal{R}_{\mu}(H)}(f(S),f(T)) \leq (2n + 1) \cdot \textsf{dist}_{\mathcal{R}_k(G)}(S,T)$ 
(where the second inequality follows from our construction of $\widehat{\sigma}$ and the fact that $f(v) \leq 2n + 1$, for all $v \in V(G)$).  
To show that $\sigma'$ is a reconfiguration sequence from $S$ to $T$ we invoke Corollary~\ref{co-vc-size}, i.e., 
for every $\tilde{Q} \in \tilde{\sigma}$ we have $|f^{-1}(\tilde{Q})| \leq k$ and therefore $f^{-1}(\tilde{Q}) \in V(\mc{R}_k(G))$. 
Note that, after deleting duplicate sets to obtain $\sigma'$, the size of the symmetric difference of consecutive sets is again exactly one.  
\end{proof}

\begin{lemma}\label{le-inequality2}
$\forall \widehat{S},\widehat{T} \in V(\mc{R}_\mu(H))$, if $\textsf{dist}_{\mathcal{R}_\mu(H)}(\widehat{S},\widehat{T})$ is 
finite then $\textsf{dist}_{\mathcal{R}_{k}(G)}(f^{-1}(\widehat{S}),f^{-1}(\widehat{T}))$ is finite 
and $\textsf{dist}_{\mathcal{R}_\mu(H)}(\widehat{S},\widehat{T}) \leq (2n + 1) \cdot \textsf{dist}_{\mathcal{R}_{k}(G)}(f^{-1}(\widehat{S}),f^{-1}(\widehat{T})) + \Oh(n^2)$. 
\end{lemma}

\begin{proof}
Let $\widehat{\sigma} = \langle \widehat{Q}_0 = \widehat{S}, \widehat{Q}_1 ,\ldots, \widehat{Q}_{t-1}, \widehat{Q}_t = \widehat{T} \rangle$ 
denote a shortest reconfiguration sequence from $\widehat{S}$ to $\widehat{T}$ in $\mathcal{R}_{\mu}(H)$. 
By Corollary~\ref{co-vc-size}, we know that $|f^{-1}(\widehat{Q}_i)| \leq k$ and therefore $f^{-1}(\widehat{Q}_i) \in V(\mathcal{R}_{k}(G))$, for all $i$. 
Consider the sequence $\sigma$ obtained from $\widehat{\sigma}$ by replacing every set 
$\widehat{Q} \in \widehat{\sigma}$ by $f^{-1}(\widehat{Q})$ and then deleting duplicate sets. Using the same arguments as in the proof of Lemma~\ref{le-inequality1}, 
we know that $\sigma$ is a reconfiguration sequence from $f^{-1}(\widehat{S})$ to $f^{-1}(\widehat{T})$ in $\mc{R}_k(G)$. 
Moreover, from Lemma~\ref{le-inequality1}, we know that 
\begin{equation}\label{eq1}
\begin{split}
\textsf{dist}_{\mathcal{R}_k(G)}(f^{-1}(\widehat{S}),f^{-1}(\widehat{T})) 
 & \leq \textsf{dist}_{\mathcal{R}_{\mu}(H)}(f(f^{-1}(\widehat{S})),f(f^{-1}(\widehat{T}))) \\
 & \leq (2n + 1) \cdot \textsf{dist}_{\mathcal{R}_k(G)}(f^{-1}(\widehat{S}),f^{-1}(\widehat{T})).
\end{split}
\end{equation}
In addition, we know that 
\begin{equation}\label{eq2}
\begin{split}
\textsf{dist}_{\mathcal{R}_\mu(H)}(\widehat{S},\widehat{T}) \leq \textsf{dist}_{\mathcal{R}_\mu(H)}(\widehat{S},f(f^{-1}(\widehat{S})))
 & + \textsf{dist}_{\mathcal{R}_\mu(H)}(f(f^{-1}(\widehat{S})),f(f^{-1}(\widehat{T}))) \\
 & + \textsf{dist}_{\mathcal{R}_\mu(H)}(f(f^{-1}(\widehat{T})),\widehat{T}).
\end{split}
\end{equation}
Combining (\ref{eq1}) and (\ref{eq2}) with the fact that $\textsf{dist}_{\mathcal{R}_\mu(H)}(\widehat{Q},f(f^{-1}(\widehat{Q})))$ is at most $|V(H)|$ 
(for any $\widehat{Q} \in \mc{R}_\mu(H)$), we get the claimed bound on $\textsf{dist}_{\mathcal{R}_\mu(H)}(\widehat{S},\widehat{T})$. 
\end{proof}

The next corollary follows from Lemma~\ref{le-inequality1} and~\ref{le-inequality2} above. 

\begin{corollary}\label{co-components}
$\forall S,T \in V(\mc{R}_k(G))$, $S$ and $T$ belong to the same connected component of $\mc{R}_k(G)$ if and only if  
$f(S)$ and $f(T)$ belong to the same connected component of $\mc{R}_\mu(H)$. 
$\forall \widehat{S},\widehat{T} \in V(\mc{R}_\mu(H))$, $\widehat{S}$ and $\widehat{T}$ belong to the same connected component 
of $\mc{R}_\mu(H)$ if and only if $f^{-1}(\widehat{S})$ and $f^{-1}(\widehat{T})$ belong to the same connected component of $\mc{R}_k(G)$. 
\end{corollary}

\begin{lemma}\label{le-minima}
Let $\mc{C}$ denote the connected component of $\mc{R}_k(G)$ containing $S^\star$ and 
let $\widehat{\mc{C}}$ denote the connected component of $\mc{R}_\mu(H)$ containing $f(S^\star)$. 
Then, $f(S^\star)$ is the unique local minima of $\widehat{\mc{C}}$. 
Moreover, $\forall T \in V(\mc{C})$, $\textsf{dist}_{\mathcal{R}_k(G)}(S^\star,T) \leq \textsf{dist}_{\mathcal{R}_{\mu}(H)}(f(S^\star),f(T)) \leq \Oh(n^4)$.
\end{lemma}

\begin{proof} 
First, recall that, for $X \subseteq V(G)$, we have $|f(X)| = (2n + 1)|X| - |S^\star \cap X|$. Moreover, 
for $\widehat{X} \subseteq V(H)$, we have $|\widehat{X}| \geq f(f^{-1}(\widehat{X}))$. 
In particular, we have $|\widehat{X}| \geq f(f^{-1}(\widehat{X})) = (2n + 1)|f^{-1}(\widehat{X})| - |S^\star \cap f^{-1}(\widehat{X})|$. 
Consider $Q \in V(\mc{C})$ and $\widehat{Q} \in \widehat{\mc{C}}$. 
From Corollary~\ref{co-components}, we know that $f(Q) \in \widehat{\mc{C}}$ and $f^{-1}(\widehat{Q}) \in V(\mc{C})$. 
Assume that there exists $Q^\star \in V(\widehat{\mc{C}})$ such that $|Q^\star| \leq |f(S^\star)| = (2n + 1)|S^\star| - |S^\star|$. 
Then, 
\begin{equation}\label{eq3}
\begin{split}
(2n + 1)|f^{-1}(Q^\star)| - |S^\star \cap f^{-1}(Q^\star)| \leq |Q^\star| 
 & \leq (2n + 1)|S^\star| - |S^\star|. 
\end{split}
\end{equation} 
Since both $S^\star$ and $f^{-1}(Q^\star)$ belong to $\mc{C}$ and $S^\star$ is a local minima in $\mc{C}$, it must be the case 
that $|f^{-1}(Q^\star)| \geq |S^\star|$. But then 
\begin{equation}\label{eq4}
\begin{split}
(2n + 1)|S^\star| - |S^\star \cap f^{-1}(Q^\star)|
 & \leq (2n + 1)|f^{-1}(Q^\star)| - |S^\star \cap f^{-1}(Q^\star)| \\
 & \leq (2n + 1)|S^\star| - |S^\star|. 
\end{split}
\end{equation} 
Consequently, we have $|S^\star \cap f^{-1}(Q^\star)| \geq |S^\star|$, which is only possible if $f^{-1}(Q^\star) = S^\star$. 
It follows that $f(S^\star)$ is the unique local minima of $\widehat{\mc{C}}$.

Finally, applying Lemma~\ref{le-inequality1}, we get $\textsf{dist}_{\mathcal{R}_k(G)}(S^\star,T) \leq \textsf{dist}_{\mathcal{R}_{\mu}(H)}(f(S^\star),f(T))$. 
Applying Lemma~\ref{le-onedown} at most $\mu$ times starting from $f(T)$, we 
get $\textsf{dist}_{\mathcal{R}_{k'}(G')}(f(S^\star),f(T)) \leq \mu \cdot |V(H)| = \Oh(n^4)$, as needed. 
\end{proof}

Theorem~\ref{th-bounded-diameter} follows from combining Lemmas~\ref{le-onedown} and~\ref{le-minima}. 
To see why, let $\mc{C}$ be a connected component of $\mc{R}_k(G)$ and let $S$ and $T$ be two vertices in $V(\mc{C})$. 
By repeated applications of Lemma~\ref{le-onedown}, we know that -- assuming $S$ is not a local minima -- there exists $S^\star \in V(\mc{C})$ such that $S^\star$ is a 
local minima and $\textsf{dist}_{\mathcal{R}_k(G)}(S,S^\star) \leq nk = \Oh(n^2)$.  
Then, applying Lemma~\ref{le-minima}, we know that $\textsf{dist}_{\mathcal{R}_k(G)}(S^\star,T) = \Oh(n^4)$, implying a 
reconfiguration sequence from $S$ to $T$ of length at most $\Oh(n^4)$.  

\subsubsection*{Proof of Lemma~\ref{le-onedown}}
We start with a few additional definitions. Given a reconfiguration sequence $\sigma$ and its corresponding edit sequence $\eta$, 
we use $\textsf{touch}(\eta)$ to denote the set of vertices touched by $\eta$. Similarly, we use $\textsf{add}(\eta)$ 
and $\textsf{rem}(\eta)$ to denote the sets of added and removed vertices, respectively. We decompose $\eta$ into two types of (maximal) \emph{blocks}. 
That is, we let $\eta = \eta_1 \cdot \eta_2 \cdot \ldots \cdot \eta_\ell$. When $\textsf{add}(\eta_i) \subseteq R$ and $\textsf{rem}(\eta_i) \subseteq L$ 
we say $\eta_i$ is a \emph{right block} (vertices are added from the right side of $G$ and deleted from the left side).  
When $\textsf{add}(\eta_i) \subseteq L$ and $\textsf{rem}(\eta_i) \subseteq R$ we say $\eta_i$ is a \emph{left block}. 
Note that blocks are monotone, i.e., every vertex is touched at most once in each block and $\eta$ alternates between blocks of different types. 
A sequence (or block) is \emph{winning} whenever $\textsf{add}(\eta) < \textsf{rem}(\eta)$, 
it is \emph{losing} when $\textsf{add}(\eta) > \textsf{rem}(\eta)$, and it is \emph{neutral} otherwise. 
Given two vertex cover $S$ and $T$ belonging to the same connected component of $\mc{R}_k(G)$, we say a shortest sequence $\eta$ between them is 
\emph{tight}, if among all shortest sequences between $S$ and $T$, $\eta$ minimizes the sum of the sizes of vertex covers in $\sigma$. 
Formally, we define a potential function, $\textsf{pot}(\sigma) = \sum_{Q \in \sigma}{|Q|}$. We say $\sigma$ (or $\eta$) is 
tight if $\sigma$ is shortest and $\textsf{pot}(\sigma)$ is minimized. 

\begin{lemma}\label{le-swap}
Let $\eta = \eta_1 \cdot \eta_2 \cdot \ldots \cdot \eta_\ell$ be a valid edit sequence starting at $S$, where each $\eta_i$ is a (maximal) block. If there exists 
$i$ such that $\eta_i$ is not winning, $\eta_{i+1}$ is winning, and $\textsf{touch}(\eta_i) \cap \textsf{touch}(\eta_{i + 1}) = \emptyset$ then 
$\eta' = \eta_1 \cdot \ldots \cdot \eta_{i+1} \cdot \eta_i \cdot \ldots \cdot \eta_\ell$ is also valid starting at $S$. 
\end{lemma}

\begin{proof}
Assume, w.l.o.g., that $\eta_i$ is a right block and $\eta_{i+1}$ is a left block. 
Let $A_1 = \textsf{add}(\eta_i)$, $D_1 = \textsf{rem}(\eta_i)$, $A_{2} = \textsf{add}(\eta_{i + 1})$, and $D_{2} = \textsf{rem}(\eta_{i+1})$. 
Since $\textsf{touch}(\eta_i) \cap \textsf{touch}(\eta_{i + 1}) = \emptyset$, we know that $D_1 \cap A_2 = D_2 \cap A_1 = \emptyset$. 
Let $Q' = \eta_1 \cdot \eta_2 \cdot \ldots \cdot \eta_{i-1}(S)$ and $Q'' = \eta_1 \cdot \eta_2 \cdot \ldots \cdot \eta_{i+1}(S)$. 
Note that before the start of $\eta_i$, all vertices in $D_1 \cup D_2$ are in the vertex cover $Q'$, i.e., $D_1 \cup D_2 \subseteq Q'$ and 
all vertices in $A_1 \cup A_2$ are not, i.e., $(A_1 \cup A_2) \cap Q' = \emptyset$.  
Hence, there are no edges between $A_1$ and $A_2$. 
Similarly, after $\eta_{i+1}$ all vertices in $A_1 \cup A_2$ are in $Q''$ and all vertices in $D_1 \cup D_2$ are not,  
and therefore there are no edges between $D_1$ and $D_2$. 

Since $\eta$ is valid and $\eta_{i+1}$ is a winning block, we know that every set in $\sigma'$ is of size at most $k$. If there exists a set $Q \in \sigma'$ 
that is not a vertex cover then some edge $uv$ is left uncovered. However, $uv$ must have one endpoint in $A_1$ 
and the other in $A_2$ or one endpoint in $D_1$ and the other in $D_2$. But we have just shown that such edges cannot exist. 
\end{proof}

A \emph{crown} is an ordered pair $(C,H)$ of subsets of vertices from $G$ that satisfies the
following criteria: (1) $C \neq \emptyset$ is an independent set of $G$, (2) $H = N(C)$, and (3) there exists a
matching $M$ on the edges connecting $C$ and $H$ such that all elements of $H$ are matched ($M$ \emph{saturates} $H$). 
$H$ is called the \emph{head} of the crown. Note that by definition $|C| \geq |H|$. 
Crowns have proved to be very useful for the design of fixed-parameter tractable and kernelization algorithms 
for the {\sc Vertex Cover} problem and many others~\cite{DBLP:journals/mst/Abu-KhzamFLS07,pcbook}.  
Given $G$, two vertex covers $S$ and $T$ of $G$ of size at most $k$, and an edit sequence $\eta$ transforming $S$ to $T$, we 
say $(C,H)$ is an \emph{$\eta$-local crown} if $(C,H)$ is a crown in $G[\textsf{touch}(\eta)]$. 
In other words, a local crown ignores all vertices that are never touched by $\eta$. 
The usefulness of local crowns is captured by the next lemma. 

\begin{lemma}\label{le-crown}
Let $S$ and $T$ be two vertex covers of $G$ of size at most $k$ and let $\eta$ be an edit sequence transforming $S$ to $T$. 
Moreover, let $(C,H)$ be an $\eta$-local crown, i.e, $C \cup H \subseteq \textsf{touch}(\eta)$. 
If $H \subseteq S$, $C \cap S = \emptyset$, $H \subseteq T$, and $C \cap T = \emptyset$ then $\eta$ is not a shortest edit sequence from $S$ to $T$.  
\end{lemma}

\begin{proof}
Note that, as $M$ saturates $H$, at least $|H|$ vertices are needed to cover edges in $G[C \cup H]$.  
In other words, for any vertex cover $Q$ of $G$, at least $|H|$ vertices of $Q$ are in $C \cup H$. 
Moreover, $(Q \setminus (C \cup H)) \cup H$ is a vertex cover of $G$ of size at most $|Q|$. 
Now consider the edit sequence $\eta'$ obtained from $\eta$ by skipping/deleting all additions and removals of vertices in $C \cup H$. 
Since $|C| \geq |H|$, $H = N(C)$, and vertices in $H$ are never touched, it follows that $\eta'$ remains valid. 
\end{proof}

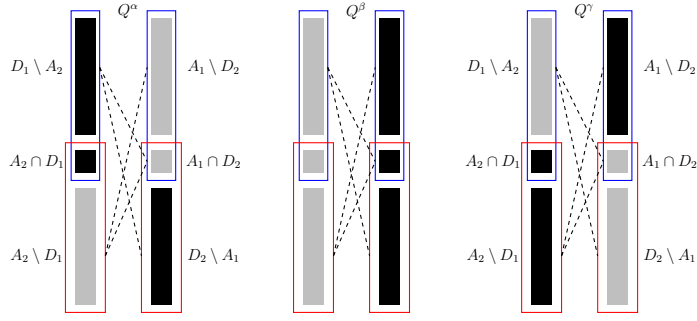
\begin{figure}
\centering
\scalebox{0.5}{\begin{tikzpicture}

\node at (1.5, 8) {\large $Q^\alpha$};
\node at (7.5, 8) {\large $Q^\beta$};
\node at (13.5, 8) {\large $Q^\gamma$};

\node at (-0.9,6.5) {\large $D_1 \setminus A_2$};
\node at (-0.9,4) {\large $A_2 \cap D_1$};
\node at (-0.9,1.5) {\large $A_2 \setminus D_1$};
\node at (3.75,6.5) {\large $A_1 \setminus D_2$};
\node at (3.75,4) {\large $A_1 \cap D_2$};
\node at (3.75,1.5) {\large $D_2 \setminus A_1$};

\node at (-0.9 + 12, 6.5) {\large $D_1 \setminus A_2$};
\node at (-0.9 + 12, 4) {\large $A_2 \cap D_1$};
\node at (-0.9 + 12, 1.5) {\large $A_2 \setminus D_1$};
\node at (3.75 + 12, 6.5) {\large $A_1 \setminus D_2$};
\node at (3.75 + 12, 4) {\large $A_1 \cap D_2$};
\node at (3.75 + 12, 1.5) {\large $D_2 \setminus A_1$};

\draw[blue,thick] (0, 3.5) rectangle (0.75, 8);
\draw[red,thick] (0 - 0.15, 0) rectangle (0.75 + 0.15, 4.5);
\draw[blue,thick] (0 + 2, 3.5) rectangle (0.75 + 2, 8);
\draw[red,thick] (0 - 0.15 + 2, 0) rectangle (0.75 + 0.15 + 2, 4.5);
\draw[black,thick,dashed] (0.75 + 0.15, 1.5) -- (2, 4);
\draw[black,thick,dashed] (0.75 + 0.15, 1.5) -- (2, 6.5);
\draw[black,thick,dashed] (0.75, 6.5) -- (2, 4);
\draw[black,thick,dashed] (0.75, 6.5) -- (2 - 0.15, 1.5);

\draw[blue,thick] (0 + 6, 3.5) rectangle (0.75 + 6, 8);
\draw[red,thick] (0 - 0.15 + 6, 0) rectangle (0.75 + 0.15 + 6, 4.5);
\draw[blue,thick] (0 + 2 + 6, 3.5) rectangle (0.75 + 2 + 6, 8);
\draw[red,thick] (0 - 0.15 + 2 + 6, 0) rectangle (0.75 + 0.15 + 2 + 6, 4.5);
\draw[black,thick,dashed] (0.75 + 0.15 + 6, 1.5) -- (2 + 6, 4);
\draw[black,thick,dashed] (0.75 + 0.15 + 6, 1.5) -- (2 + 6, 6.5);
\draw[black,thick,dashed] (0.75 + 6, 6.5) -- (2 + 6, 4);
\draw[black,thick,dashed] (0.75 + 6, 6.5) -- (2 - 0.15 + 6, 1.5);

\draw[blue,thick] (0 + 12, 3.5) rectangle (0.75 + 12, 8);
\draw[red,thick] (0 - 0.15 + 12, 0) rectangle (0.75 + 0.15 + 12, 4.5);
\draw[blue,thick] (0 + 2 + 12, 3.5) rectangle (0.75 + 2 + 12, 8);
\draw[red,thick] (0 - 0.15 + 2 + 12, 0) rectangle (0.75 + 0.15 + 2 + 12, 4.5);
\draw[black,thick,dashed] (0.75 + 0.15 + 12, 1.5) -- (2 + 12, 4);
\draw[black,thick,dashed] (0.75 + 0.15 + 12, 1.5) -- (2 + 12, 6.5);
\draw[black,thick,dashed] (0.75 + 12, 6.5) -- (2 + 12, 4);
\draw[black,thick,dashed] (0.75 + 12, 6.5) -- (2 - 0.15 + 12, 1.5);

\fill[black] (0 + 0.1, 4.5 + 0.2) rectangle (0.75 - 0.1, 8 - 0.2);
\fill[black] (0 + 0.1, 3.5 + 0.2) rectangle (0.75 - 0.1, 4.5 - 0.2);
\fill[gray!50] (0 + 0.1, 0 + 0.2) rectangle (0.75 - 0.1, 3.5 - 0.2);
\fill[gray!50] (0 + 0.1 + 2, 4.5 + 0.2) rectangle (0.75 - 0.1 + 2, 8 - 0.2);
\fill[gray!50] (0 + 0.1 + 2, 3.5 + 0.2) rectangle (0.75 - 0.1 + 2, 4.5 - 0.2);
\fill[black] (0 + 0.1 + 2, 0 + 0.2) rectangle (0.75 - 0.1 + 2, 3.5 - 0.2);

\fill[gray!50] (0 + 0.1 + 6, 4.5 + 0.2) rectangle (0.75 - 0.1 + 6, 8 - 0.2);
\fill[gray!50] (0 + 0.1 + 6, 3.5 + 0.2) rectangle (0.75 - 0.1 + 6, 4.5 - 0.2);
\fill[gray!50] (0 + 0.1 + 6, 0 + 0.2) rectangle (0.75 - 0.1 + 6, 3.5 - 0.2);
\fill[black] (0 + 0.1 + 2 + 6, 4.5 + 0.2) rectangle (0.75 - 0.1 + 2 + 6, 8 - 0.2);
\fill[black] (0 + 0.1 + 2 + 6, 3.5 + 0.2) rectangle (0.75 - 0.1 + 2 + 6, 4.5 - 0.2);
\fill[black] (0 + 0.1 + 2 + 6, 0 + 0.2) rectangle (0.75 - 0.1 + 2 + 6, 3.5 - 0.2);

\fill[gray!50] (0 + 0.1 + 12, 4.5 + 0.2) rectangle (0.75 - 0.1 + 12, 8 - 0.2);
\fill[black] (0 + 0.1 + 12, 3.5 + 0.2) rectangle (0.75 - 0.1 + 12, 4.5 - 0.2);
\fill[black] (0 + 0.1 + 12, 0 + 0.2) rectangle (0.75 - 0.1 + 12, 3.5 - 0.2);
\fill[black] (0 + 0.1 + 2 + 12, 4.5 + 0.2) rectangle (0.75 - 0.1 + 2 + 12, 8 - 0.2);
\fill[gray!50] (0 + 0.1 + 2 + 12, 3.5 + 0.2) rectangle (0.75 - 0.1 + 2 + 12, 4.5 - 0.2);
\fill[gray!50] (0 + 0.1 + 2 + 12, 0 + 0.2) rectangle (0.75 - 0.1 + 2 + 12, 3.5 - 0.2);

\end{tikzpicture}}
\caption{(Parts of) vertex covers $Q^\alpha$ (left), $Q^\beta$ (middle), and $Q^\gamma$ (right) shown in black (forbidden edges drawn as dotted lines).} 
\label{fig-states-and-forbid}
\end{figure}

We are now ready to prove our main lemma. 

\begin{lemma}
Let $G$ be a bipartite graph and $S$ be a minimal vertex cover of $G$ of size at most $k$. Let $\eta$ be a valid tight winning sequence 
starting from $S$. Then, $\eta$ consists of one (either left or right) block, i.e., $\eta$ is monotone. 
\end{lemma}

\begin{proof}
Assume otherwise. Let $\eta = \eta_1 \cdot \ldots \cdot \eta_\ell$ and let $\sigma = \sigma_1 \cdot \ldots \cdot \sigma_\ell$ be the corresponding sequence of vertex covers. 
Then, there exists at least one non-winning block followed by a winning block, say $\eta_i$ and $\eta_{i+1}$. 
If $\textsf{touch}(\eta_i) \cap \textsf{touch}(\eta_{i + 1}) = \emptyset$ then, by Lemma~\ref{le-swap}, we can swap $\eta_i$ and $\eta_{i+1}$ to obtain sequences 
$\eta' = \eta_1 \cdot \ldots \cdot \eta_{i+1} \cdot \eta_i \cdot \ldots \cdot \eta_\ell$ and 
$\sigma' = \sigma_1 \cdot \ldots \cdot \sigma'_{i+1} \cdot \sigma'_i \cdot \ldots \cdot \sigma_\ell$, which are also valid. 
However, as $\eta_{i+1}$ is winning and $\eta_i$ is not, $\eta_1 \cdot \ldots \cdot \eta_{i-1} \cdot \eta_{i+1}(S) < \eta_1 \cdot \ldots \cdot \eta_{i-1} \cdot \eta_{i}(S)$ and 
therefore $\sum_{Q' \in \sigma'_{i+1} \cdot \sigma'_i}{|Q'|} < \sum_{Q \in \sigma_{i} \cdot \sigma_{i+1}}{|Q|}$ (the size of all other vertex cover does not increase). 
Hence, $\textsf{pot}(\sigma') = \sum_{Q' \in \sigma'}{|Q'|} < \textsf{pot}(\sigma) = \sum_{Q \in \sigma}{|Q|}$, contradicting the fact that $\eta$ is tight. 

From now on we assume, w.l.o.g., that $\textsf{touch}(\eta_i) \cap \textsf{touch}(\eta_{i + 1}) \neq \emptyset$, $\eta_i$ is a right block, and $\eta_{i+1}$ is a left block. 
Let $A_1 = \textsf{add}(\eta_i)$, $D_1 = \textsf{rem}(\eta_i)$, $A_{2} = \textsf{add}(\eta_{i + 1})$, and $D_{2} = \textsf{rem}(\eta_{i+1})$. 
Let $Q^\alpha = \eta_1 \cdot \eta_2 \cdot \ldots \cdot \eta_{i-1}(S)$, $Q^\beta = \eta_1 \cdot \eta_2 \cdot \ldots \cdot \eta_{i}(S)$, 
and $Q^\gamma = \eta_1 \cdot \eta_2 \cdot \ldots \cdot \eta_{i+1}(S)$ (see Figure~\ref{fig-states-and-forbid}). 

\begin{claim}\label{cl-no-edges}
There are no edges between $A_1 \cap D_2$ and $D_1 \setminus A_2$, or $A_1 \cap D_2$ and $A_2 \setminus D_1$. 
Similarly, there are no edges between $D_1 \setminus A_2$ and $D_2 \setminus A_1$ or $A_2 \setminus D_1$ and $A_1 \setminus D_2$.
\end{claim}

\begin{proof}
The fact that $((A_1 \setminus D_2) \cup (A_2 \setminus D_1)) \cap Q^\alpha = \emptyset$ and 
$((D_1 \setminus A_2) \cup (D_2 \setminus A_1)) \cap Q^\gamma = \emptyset$ implies that there are no edges 
between $A_1 \setminus D_2$ and $A_2 \setminus D_1$ or between  $D_1 \setminus A_2$ and $D_2 \setminus A_1$. 
As $(A_1 \cap D_2) \cup Q^\alpha = \emptyset$ and $(A_2 \setminus D_1) \cap Q^\alpha = \emptyset$, there can be no edges between 
$A_1 \cap D_2$ and $A_2 \setminus D_1$. 
Finally, as $(A_1 \cap D_2) \cup Q^\gamma = \emptyset$ and $(D_1 \setminus D_2) \cap Q^\gamma = \emptyset$, 
there can be no edges between $A_1 \cap D_2$ and $D_1 \setminus A_2$ (Figure~\ref{fig-states-and-forbid}). 
\end{proof}

Now consider the sequence $\eta' = \eta_1 \cdot \ldots \cdot \eta'_{i+1} \cdot \eta'_i \cdot \ldots \cdot \eta_\ell$ obtained from $\eta$ by swapping $\eta_i$ and $\eta_{i+1}$ and 
deleting all additions and removals of vertices in $(A_1 \cap D_2) \cup (A_2 \cap D_1)$. 
Let $\sigma' = \sigma_1 \cdot \ldots \cdot \sigma'_{i+1} \cdot \sigma'_i \cdot \ldots \cdot \sigma_\ell$ be the corresponding sequence of vertex covers. 

\begin{claim}
Every set $Q'$ in $\sigma'_{i+1} \cdot \sigma'_i$ is a vertex cover of $G$. 
\end{claim}

\begin{proof} 
By Claim~\ref{cl-no-edges}, we know that $N(A_1 \cap D_2) \subseteq A_2 \cap D_1$. 
Since vertices of $A_2 \cap D_1$ are never removed, all edges with one endpoint in $A_2 \cap D_1$ are covered. 
Again by Claim~\ref{cl-no-edges}, for each remaining edge $uv$, we either have $u \in A_1 \setminus D_2$ and $v \in D_1 \setminus A_2$ or 
$u \in A_2 \setminus D_1$ and $v \in D_2 \setminus A_1$. As the original sequence is valid, we know that such edges are covered. 
\end{proof}

\begin{claim}\label{cl-old-new}
Every set $Q'$ in $\sigma'_{i+1}$ can be associated with a set $Q$ in $\sigma_{i+1}$ such that 
$|Q'| = |Q| + |D_1| - |A_1| + |R_{\textsf{out}}| - |L_{\textsf{in}}|$, where $R_{\textsf{out}} \subseteq A_1 \cap D_2$ and $L_{\textsf{in}} \subseteq A_2 \cap D_1$. 
Similarly, every set $Q'$ in $\sigma'_i$ can be associated with a set $Q$ in $\sigma_{i}$ such that 
$|Q'| = |Q| + |A_2| - |D_2| + |R'_{\textsf{out}}| - |L'_{\textsf{in}}|$, where $R'_{\textsf{out}} \subseteq A_1 \cap D_2$ and $L'_{\textsf{in}} \subseteq A_2 \cap D_1$.
\end{claim}

\begin{proof}
Let $Q'$ be a set in $\sigma'_{i+1}$. We associate $Q'$ with $Q$ in $\sigma_{i+1}$ such that 
$Q \cap ((A_2 \setminus D_1) \cup (D_2 \setminus A_1)) = Q' \cap ((A_2 \setminus D_1) \cup (D_2 \setminus A_1))$. 
Since we only delete additions and removals of vertices in $(A_1 \cap D_2) \cup (A_2 \cap D_1)$, such a set $Q$ must exist. 
Given that $\eta_i$ and $\eta_{i+1}$ are swapped in $\eta'$, it follows that $|Q'| = |Q| + |D_1| - |A_1| + |R_{\textsf{out}}| - |L_{\textsf{in}}|$; 
vertices in $D_1$ are not removed, vertices in $A_1$ are not yet added, some set of vertices $R_{\textsf{out}} \subseteq A_1 \cap D_2$ is possibly added, 
and some set of vertices $L_{\textsf{in}} \subseteq A_2 \cap D_1$ is possibly removed (Figure~\ref{fig-caps}). 

Let $Q'$ be a set in $\sigma'_{i}$. We associate $Q'$ with $Q$ in $\sigma_{i}$ such that 
$Q \cap ((A_1 \setminus D_2) \cup (D_1 \setminus A_2)) = Q' \cap ((A_1 \setminus D_2) \cup (D_1 \setminus A_2))$. 
It follows that $|Q'| = |Q| + |A_2| - |D_2| + |R'_{\textsf{out}}| - |L'_{\textsf{in}}|$; 
vertices in $A_2$ are already added, vertices in $D_2$ are removed, some set of vertices $R'_{\textsf{out}} \subseteq A_1 \cap D_2$ is possibly added, 
and some set of vertices $L'_{\textsf{in}} \subseteq A_2 \cap D_1$ is possibly removed (Figure~\ref{fig-caps}). 
\end{proof}

\begin{figure}
\centering
\scalebox{0.5}{\begin{tikzpicture}


\node at (1.5,8.5) {\large $Q' \in \sigma'_{i+1}$};
\node at (7.5,8.5) {\large $Q \in \sigma_{i+1}$};
\node at (16.5,8.5) {\large $Q' \in \sigma'_{i}$};
\node at (22.5,8.5) {\large $Q \in \sigma_{i}$};

\node at (-0.8,6.5) {\large $D_1 \setminus A_2$};
\node at (-0.8,4) {\large $A_2 \cap D_1$};
\node at (-0.8,1.5) {\large $A_2 \setminus D_1$};
\node at (3.5,6.5) {\large $A_1 \setminus D_2$};
\node at (3.5,4) {\large $A_1 \cap D_2$};
\node at (3.5,1.5) {\large $D_2 \setminus A_1$};

\node at (-0.8 + 15, 6.5) {\large $D_1 \setminus A_2$};
\node at (-0.8 + 15, 4) {\large $A_2 \cap D_1$};
\node at (-0.8 + 15, 1.5) {\large $A_2 \setminus D_1$};
\node at (3.5 + 15, 6.5) {\large $A_1 \setminus D_2$};
\node at (3.5 + 15, 4) {\large $A_1 \cap D_2$};
\node at (3.5 + 15, 1.5) {\large $D_2 \setminus A_1$};4

\draw[blue,thick] (0, 5) rectangle (0.75, 8);
\draw[blue,thick] (0 - 0.05, 3.5 - 0.05) rectangle (0.75 + 0.05, 4.5 + 0.05);
\draw[red,thick] (0, 3.5) rectangle (0.75, 4.5);
\draw[red,thick] (0, 0) rectangle (0.75, 3);

\draw[blue,thick] (2, 5) rectangle (2.75, 8);
\draw[blue,thick] (2 - 0.05, 3.5 - 0.05) rectangle (2.75 + 0.05, 4.5 + 0.05);
\draw[red,thick] (2, 3.5) rectangle (2.75, 4.5);
\draw[red,thick] (2, 0) rectangle (2.75, 3);

\draw[black,thick,dashed] (0.75, 1.5) -- (2, 4);
\draw[black,thick,dashed] (0.75, 1.5) -- (2, 6.5);
\draw[black,thick,dashed] (0.75, 6.5) -- (2, 4);
\draw[black,thick,dashed] (0.75, 6.5) -- (2, 1.5);

\fill[black] (0 + 0.1, 5 + 0.1) rectangle (0.75 - 0.1, 8 - 0.1);
\fill[black] (0 + 0.1, 3.5 + 0.1) rectangle (0.75 - 0.1, 4.5 - 0.1);
\fill[gray!50] (0 + 0.1, 1.6) rectangle (0.75 - 0.1, 3 - 0.1);
\fill[black] (0 + 0.1, 0 + 0.1) rectangle (0.75 - 0.1, 1.5 - 0.1);

\fill[gray!50] (0 + 0.1 + 2, 5 + 0.1) rectangle (0.75 - 0.1 + 2, 8 - 0.1);
\fill[gray!50] (0 + 0.1 + 2, 3.5 + 0.1) rectangle (0.75 - 0.1 + 2, 4.5 - 0.1);
\fill[gray!50] (0 + 0.1 + 2, 1.6) rectangle (0.75 - 0.1 + 2, 3 - 0.1);
\fill[black] (0 + 0.1 + 2, 0 + 0.1) rectangle (0.75 - 0.1 + 2, 1.5 - 0.1);
\draw[blue,thick] (0 + 6, 5) rectangle (0.75 + 6, 8);
\draw[blue,thick] (0 - 0.05 + 6, 3.5 - 0.05) rectangle (0.75 + 0.05 + 6, 4.5 + 0.05);
\draw[red,thick] (0 + 6, 3.5) rectangle (0.75 + 6, 4.5);
\draw[red,thick] (0 + 6, 0) rectangle (0.75 + 6, 3);

\draw[blue,thick] (2 + 6, 5) rectangle (2.75 + 6, 8);
\draw[blue,thick] (2 - 0.05 + 6, 3.5 - 0.05) rectangle (2.75 + 0.05 + 6, 4.5 + 0.05);
\draw[red,thick] (2 + 6, 3.5) rectangle (2.75 + 6, 4.5);
\draw[red,thick] (2 + 6, 0) rectangle (2.75 + 6, 3);

\draw[black,thick,dashed] (0.75 + 6, 1.5) -- (2 + 6, 4);
\draw[black,thick,dashed] (0.75 + 6, 1.5) -- (2 + 6, 6.5);
\draw[black,thick,dashed] (0.75 + 6, 6.5) -- (2 + 6, 4);
\draw[black,thick,dashed] (0.75 + 6, 6.5) -- (2 + 6, 1.5);

\fill[gray!50] (0 + 0.1 + 6, 5 + 0.1) rectangle (0.75 - 0.1 + 6, 8 - 0.1);
\fill[black] (0 + 0.1 + 6, 4.05) rectangle (0.75 - 0.1 + 6, 4.5 - 0.1);
\fill[gray!50] (0 + 0.1 + 6, 3.5 + 0.1) rectangle (0.75 - 0.1 + 6, 3.95);
\fill[gray!50] (0 + 0.1 + 6, 1.6) rectangle (0.75 - 0.1 + 6, 3 - 0.1);
\fill[black] (0 + 0.1 + 6, 0 + 0.1) rectangle (0.75 - 0.1 + 6, 1.5 - 0.1);

\fill[black] (0 + 0.1 + 2 + 6, 5 + 0.1) rectangle (0.75 - 0.1 + 2 + 6, 8 - 0.1);
\fill[black] (0 + 0.1 + 2 + 6, 4.05) rectangle (0.75 - 0.1 + 2 +6, 4.5 - 0.1);
\fill[gray!50] (0 + 0.1 + 2 + 6, 3.5 + 0.1) rectangle (0.75 - 0.1 + 2 + 6, 3.95);
\fill[gray!50] (0 + 0.1 + 2 + 6, 1.6) rectangle (0.75 - 0.1 + 2 + 6, 3 - 0.1);
\fill[black] (0 + 0.1 + 2 + 6, 0 + 0.1) rectangle (0.75 - 0.1 + 2 + 6, 1.5 - 0.1);
\draw[blue,thick] (0 + 15, 5) rectangle (0.75 + 15, 8);
\draw[blue,thick] (0 - 0.05 + 15, 3.5 - 0.05) rectangle (0.75 + 0.05 + 15, 4.5 + 0.05);
\draw[red,thick] (0 + 15, 3.5) rectangle (0.75 + 15, 4.5);
\draw[red,thick] (0 + 15, 0) rectangle (0.75 + 15, 3);

\draw[blue,thick] (2 + 15, 5) rectangle (2.75 + 15, 8);
\draw[blue,thick] (2 - 0.05 + 15, 3.5 - 0.05) rectangle (2.75 + 0.05 + 15, 4.5 + 0.05);
\draw[red,thick] (2 + 15, 3.5) rectangle (2.75 + 15, 4.5);
\draw[red,thick] (2 + 15, 0) rectangle (2.75 + 15, 3);

\draw[black,thick,dashed] (0.75 + 15, 1.5) -- (2 + 15, 4);
\draw[black,thick,dashed] (0.75 + 15, 1.5) -- (2 + 15, 6.5);
\draw[black,thick,dashed] (0.75 + 15, 6.5) -- (2 + 15, 4);
\draw[black,thick,dashed] (0.75 + 15, 6.5) -- (2 + 15, 1.5);
\draw[blue,thick] (0 + 21, 5) rectangle (0.75 + 21, 8);
\draw[blue,thick] (0 - 0.05 + 21, 3.5 - 0.05) rectangle (0.75 + 0.05 + 21, 4.5 + 0.05);
\draw[red,thick] (0 + 21, 3.5) rectangle (0.75 + 21, 4.5);
\draw[red,thick] (0 + 21, 0) rectangle (0.75 + 21, 3);

\draw[blue,thick] (2 + 21, 5) rectangle (2.75 + 21, 8);
\draw[blue,thick] (2 - 0.05 + 21, 3.5 - 0.05) rectangle (2.75 + 0.05 + 21, 4.5 + 0.05);
\draw[red,thick] (2 + 21, 3.5) rectangle (2.75 + 21, 4.5);
\draw[red,thick] (2 + 21, 0) rectangle (2.75 + 21, 3);

\draw[black,thick,dashed] (0.75 + 21, 1.5) -- (2 + 21, 4);
\draw[black,thick,dashed] (0.75 + 21, 1.5) -- (2 + 21, 6.5);
\draw[black,thick,dashed] (0.75 + 21, 6.5) -- (2 + 21, 4);
\draw[black,thick,dashed] (0.75 + 21, 6.5) -- (2 + 21, 1.5);
\fill[black] (0 + 15 + 0.1, 6.6) rectangle (0.75 + 15 - 0.1, 8 - 0.1);
\fill[gray!50] (0 + 15 + 0.1, 5 + 0.1) rectangle (0.75 + 15 - 0.1, 6.5);
\fill[black] (0 + 17 + 0.1, 6.6) rectangle (0.75 + 17 - 0.1, 8 - 0.1);
\fill[gray!50] (0 + 17 + 0.1, 5 + 0.1) rectangle (0.75 + 17 - 0.1, 6.5);
\fill[black] (0 + 21 + 0.1, 6.6) rectangle (0.75 + 21 - 0.1, 8 - 0.1);
\fill[gray!50] (0 + 21 + 0.1, 5 + 0.1) rectangle (0.75 + 21 - 0.1, 6.5);
\fill[black] (0 + 23 + 0.1, 6.6) rectangle (0.75 + 23 - 0.1, 8 - 0.1);
\fill[gray!50] (0 + 23 + 0.1, 5 + 0.1) rectangle (0.75 + 23 - 0.1, 6.5);

\fill[black] (0 + 15 + 0.1, 3.5 + 0.1) rectangle (0.75 + 15 - 0.1, 4.5 - 0.1);
\fill[gray!50] (2 + 15 + 0.1, 3.5 + 0.1) rectangle (2.75 + 15 - 0.1, 4.5 - 0.1);
\fill[black] (0 + 21 + 0.1, 4.05) rectangle (0.75 + 21 - 0.1, 4.5 - 0.1);
\fill[gray!50] (0 + 21 + 0.1, 3.6) rectangle (0.75 + 21 - 0.1, 3.95);
\fill[black] (2 + 21 + 0.1, 4.05) rectangle (0.75 + 23 - 0.1, 4.5 - 0.1);
\fill[gray!50] (2 + 21 + 0.1, 3.6) rectangle (0.75 + 23 - 0.1, 3.95);

\fill[black] (0 + 15 + 0.1, 0 + 0.1) rectangle (0.75 + 15 - 0.1, 3 - 0.1);
\fill[gray!50] (2 + 15 + 0.1, 0 + 0.1) rectangle (2.75 + 15 - 0.1, 3 - 0.1);

\fill[gray!50] (0 + 21 + 0.1, 0 + 0.1) rectangle (0.75 + 21 - 0.1, 3 - 0.1);
\fill[black] (2 + 21 + 0.1, 0 + 0.1) rectangle (2.75 + 21 - 0.1, 3 - 0.1);
\end{tikzpicture}}
\caption{$|Q'| = |Q| + |D_1| - |A_1| + |R_{\textsf{out}}| - |L_{\textsf{in}}|$ (left, forbidden edges shown as dotted lines).
$|Q'| = |Q| + |A_2| - |D_2| + |R'_{\textsf{out}}| - |L'_{\textsf{in}}|$ (right).} 
\label{fig-caps}
\end{figure}
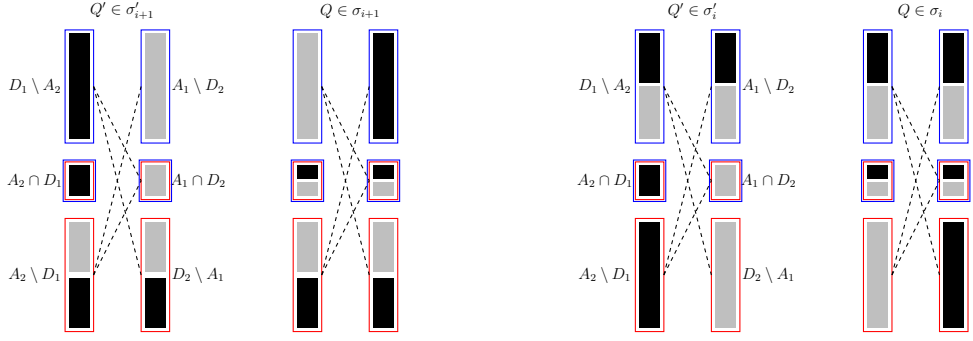

Recall that $\eta_i$ is a non-winning block and $\eta_{i+1}$ is a winning block. 
Therefore, $|A_1| - |D_1| \geq 0$ (or $|D_1| - |A_1| \leq 0$) and $|A_2| - |D_2| \leq 0$. 
Consequently, if $|R_{\textsf{out}}| - |L_{\textsf{in}}| \leq 0$ and $|R'_{\textsf{out}}| - |L'_{\textsf{in}}| \leq 0$ 
for all sets in $\sigma'_{i+1} \cdot \sigma'_i$ then $\eta'$ is a shorter sequence from $S$ to $T$, again contradicting our choice of $\eta$. 
So there exists at least one set such that either $|R_{\textsf{out}}| - |L_{\textsf{in}}| > 0$ (or $|R'_{\textsf{out}}| - |L'_{\textsf{in}}| > 0$). 
But then there exists a pair $(C = R_{\textsf{out}}, H = L_{\textsf{in}})$ such that $|C| > |H|$, $N(C) \subseteq H$ 
(by Claim~\ref{cl-no-edges}), 
$H \subseteq A_1 \cap D_2 \subseteq Q^\alpha$, $C \cap Q^\alpha = \emptyset$, $H \subseteq Q^\gamma$, and $C \cap Q^\gamma = \emptyset$. 
We claim that there exists $(C' \subseteq C, H' \subseteq H)$ such that $(C',H')$ is an 
($\eta_{i} \cdot \eta_{i+1}$)-local crown. Applying Lemma~\ref{le-crown} to the local crown completes the proof; as it implies that $\eta_{i} \cdot \eta_{i+1}$ is not the 
shortest possible sequence from $Q^\alpha$ to $Q^\gamma$.   

We conclude the proof by constructing $(C',H')$. First, we set $C' = C$ and $H' = H$. Next, we delete all vertices in 
$H'$ with no neighbors in $C'$ so that $N(C') = H'$ (note that $|C|' > |H'|$ still holds). 
The remaining condition for $(C',H')$ to satisfy is for $G[C' \cup H']$ to have a matching which saturates $H'$.
Hall's Marriage Theorem~\cite{H87} states that such a saturating matching exists if and only if
for every subset $W$ of $H'$, $|W| \leq |N_{G[C' \cup H']}(W)|$. 
By a simple application of Hall's theorem, if
no saturating matching exists then there exists a
subgraph $Z$ of $G[C' \cup H']$ such that $|V(Z) \cap C'| < |V(Z) \cap H'|$.
By repeatedly deleting such subgraphs, we eventually reach a pair $(C',H')$ which satisfies all the required properties.
Since $|C| > |H|$, such a pair is guaranteed to exist as otherwise every subset $W$ of $H$ would satisfy
$|W| > |N_{G[C \cup H]}(W)|$ and hence $|H| > |C|$, a contradiction. 
\end{proof}

\section{PSPACE-completeness under the token sliding model}\label{sec-pspace}
We now switch from the ``vertex cover view'' to the ``independent set view'' of the problem.  
Given a graph $G$ and an integer $k$, the reconfiguration graph $\mc{S}_k(G)$ is the graph whose nodes are  
independent sets of $G$ of size exactly $k$. Two independent sets $I$ and $J$ are adjacent in $\mc{S}_k(G)$ 
if we can transform one into the other by sliding a token along an edge. 
More formally, $I$ and $J$ are adjacent in $\mc{S}_k(G)$ if $J \setminus I = \{u\}$, 
$I \setminus J$ = \{v\}, and $uv$ is an edge of G. 
The {\sc Token Sliding} problems asks, given a graph $G$ and two independent sets $I$ and $J$ of size $k$, whether 
$\textsf{dist}_{\mc{S}_k(G)}(I,J)$ is finite. 

\begin{figure}
\centering
    \begin{tikzpicture}[scale=.35,every node/.style={circle},inner/.style={circle},outer/.style={circle}]
        \node [inner, circle, fill=white!50, draw=black, inner sep=1pt] (s1)  {\small $h_1$};
        \node [inner, circle, fill=white!50, draw=black, inner sep=1pt, right=1.5cm of s1] (h1) {\small $f_1$};
		\node [inner, circle, fill=white!50, draw=black, inner sep=1pt, right=1.5cm of h1] (h2) {\small $f_2$};
		\node [inner, circle, fill=white!50, draw=black, inner sep=1pt, right=1.5cm of h2] (h3) {\small $f_3$};
		\node [inner, circle, fill=white!50, draw=black, inner sep=1pt, right=1.5cm of h3] (h4) {\small $f_4$};
		\node [inner, circle, fill=white!50, draw=black, inner sep=1pt, right=1.5cm of h4] (s2) {\small $h_2$};
		\node [inner, circle, fill=white!50, draw=black, inner sep=1pt, above=1.0cm of h2] (h5) {\small $f_5$};
		\node [inner, circle, fill=white!50, draw=black, inner sep=1pt, above=1.0cm of h3] (h6) {\small $f_6$};
		\node [inner, circle, fill=white!50, draw=black, inner sep=1pt, below=1.0cm of h2] (h7) {\small $f_7$};
		\node [inner, circle, fill=white!50, draw=black, inner sep=1pt, below=1.0cm of h3] (h8) {\small $f_8$};
		
        \foreach \from/\to in {s1/h1,h1/h2,h2/h3,h3/h4,h4/s2,h1/h5,h1/h7,h5/h6,h7/h8,h4/h8,h4/h6}
        \draw (\from) -- (\to);
		\draw (s1) to [out=50,in=130] (h6);
		\draw (s1) to [out=45,in=150] (h3);
		\draw (s1) to [out=-50,in=-130] (h8);
		\draw (s2) to [out=130,in=50] (h5);
		\draw (s2) to [out=150,in=45] (h2);
		\draw (s2) to [out=-130,in=-50] (h7);
	\end{tikzpicture}
\caption{The bipartite extended cage graph $H$.}
\label{fig-cage}
\end{figure}
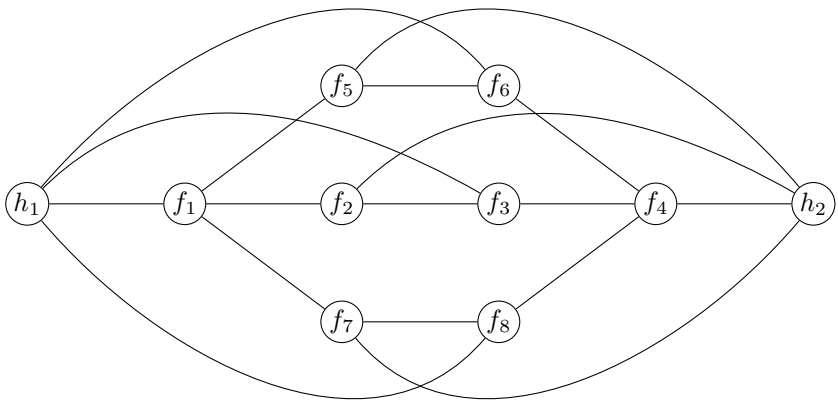

We show that {\sc Token Sliding} is \PSPACE-complete by a reduction from a variant of the {\sc Word Reconfiguration} problem.
Given a pair $W = (\Sigma, A)$, where $\Sigma = \{\sigma_1, \sigma_2, \ldots, \sigma_{|\Sigma|}\}$ is a set of symbols and 
$A = \{(\sigma_{x},\sigma_{x'}), \ldots, (\sigma_{y},\sigma_{y'})\}\subseteq \Sigma \times \Sigma = \Sigma^2$ is a binary relation between symbols, a string 
over $\Sigma$ is a \emph{word} if every two consecutive symbols are in the relation $A$. 
If one looks at $W$ as a directed graph (possibly with loops), a string is a word if and only
if it is a walk in $W$. The \PSPACE-complete {\sc Word Reconfiguration} problem~\cite{DBLP:journals/corr/Wrochna14} asks whether
two given words, $w_s$ and $w_t$, of equal length $n$ can be transformed into one another 
by changing one symbol at a time so that all intermediary strings are also words. A crucial observation 
about the {\sc Word Reconfiguration} problem is that one can actually treat even/odd positions in words independently. 
Formally, we introduce the {\sc Even/Odd Word Reconfiguration} problem, where instead of changing symbols one at a time, 
we allow bulk changes as long as they occur at only even or odd positions. 
The next proposition follows from the observation above and the \PSPACE-completeness of the {\sc Word Reconfiguration} problem. 

\begin{proposition}\label{pr-even-odd}
{\sc Even/Odd Word Reconfiguration} is \PSPACE-complete. 
\end{proposition}

Let $\{f_1, f_2, f_3, f_4, f_5, f_6, f_7, f_8\}$ and $\{h_1, h_2\}$ be two sets of vertices. 
We let $H$ be the graph depicted in Figure~\ref{fig-cage}. 
Formally, we have $V(H) = \{f_1, f_2, f_3, f_4, f_5, f_6, f_7, f_8\} \cup \{h_1, h_2\}$ 
and $E(H)$ consists of $h_1f_1$, $h_1f_3$, $h_1f_6$, $h_1f_8$, $h_2f_2$, $h_2f_4$, $h_2f_5$, $h_2f_7$, and 
three internally vertex-disjoint paths between $f_1$ and $f_4$, namely $P_1 = \{f_1, f_2, f_3, f_4\}$,    
$P_2 = \{f_1, f_5, f_6, f_4\}$, and $P_3 = \{f_1, f_7, f_8, f_4\}$. 
We call $F = H[V(H) \setminus \{h_1, h_2\}]$ a \emph{cage (graph)} and we call $H$ an \emph{extended cage (graph)}.  
In the remainder of this section, when we say ``add a cage between two vertices $u$ and $v$'' we mean creating a new extended cage 
and identifying $u$ with $h_1$ and $v$ with $h_2$. 

\begin{proposition}\label{pr-remain-bip}
Let $B$ be a bipartite graph with bipartition $(L,R)$. 
Let $B'$ be the graph obtained from $B$ by adding a cage between $u \in L$ and $v \in R$. 
Then, $B'$ is bipartite. 
\end{proposition}

\begin{proof}
Consider the following bipartition of $V(B')$. 
Let $L' = L \cup \{f_2, f_4, f_5, f_7\}$ and let $R' = R \cup \{f_1, f_3, f_6, f_8\}$. 
Since there are no edges between vertices in $L'$ and no edges between vertices in $R'$, the proposition follows. 
\end{proof}

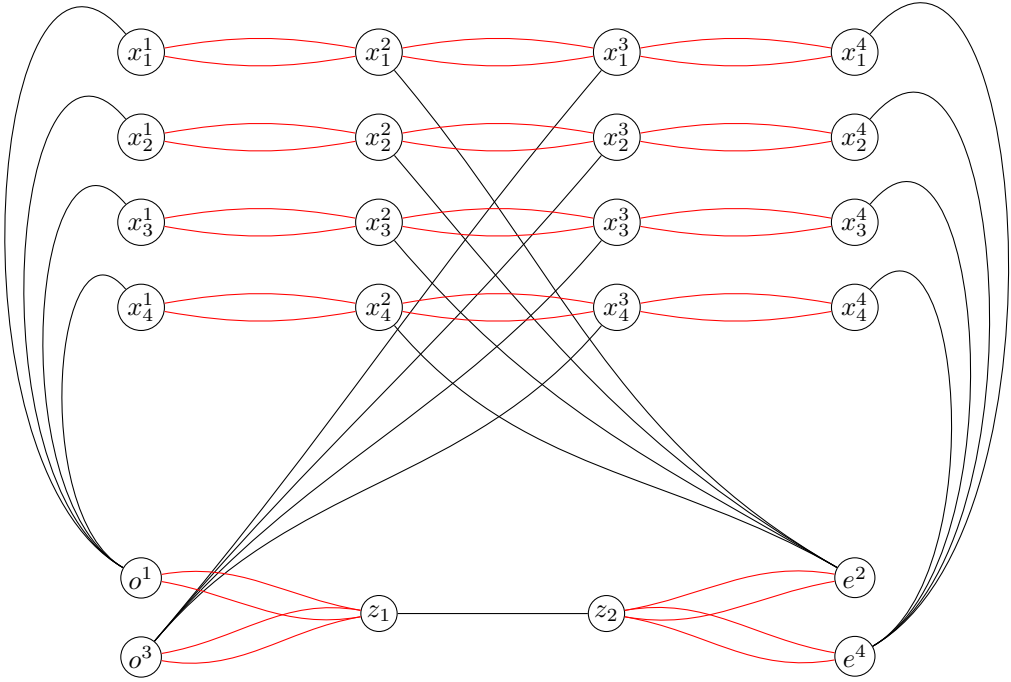
\begin{figure}[ht]
\centering
    \begin{tikzpicture}[scale=.35, auto=left, remember picture ,every node/.style={circle},inner/.style={circle},outer/.style={circle}]
        \node [inner, circle, fill=white!50, draw=black, inner sep=1pt] (x11)  {\small $x^1_1$};
        \node [inner, circle, fill=white!50, draw=black, inner sep=1pt, below=0.5cm of x11] (x12) {\small $x^1_2$};
		\node [inner, circle, fill=white!50, draw=black, inner sep=1pt, below=0.5cm of x12] (x13) {\small $x^1_3$};
		\node [inner, circle, fill=white!50, draw=black, inner sep=1pt, below=0.5cm of x13] (x14) {\small $x^1_4$};

        \node [inner, circle, fill=white!50, draw=black, inner sep=1pt, right=2.5cm of x11] (x21) {\small $x^2_1$};
        \node [inner, circle, fill=white!50, draw=black, inner sep=1pt, below=0.5cm of x21] (x22) {\small $x^2_2$};
		\node [inner, circle, fill=white!50, draw=black, inner sep=1pt, below=0.5cm of x22] (x23) {\small $x^2_3$};
		\node [inner, circle, fill=white!50, draw=black, inner sep=1pt, below=0.5cm of x23] (x24) {\small $x^2_4$};

        \node [inner, circle, fill=white!50, draw=black, inner sep=1pt, right=2.5cm of x21] (x31) {\small $x^3_1$};
        \node [inner, circle, fill=white!50, draw=black, inner sep=1pt, below=0.5cm of x31] (x32) {\small $x^3_2$};
		\node [inner, circle, fill=white!50, draw=black, inner sep=1pt, below=0.5cm of x32] (x33) {\small $x^3_3$};
		\node [inner, circle, fill=white!50, draw=black, inner sep=1pt, below=0.5cm of x33] (x34) {\small $x^3_4$};
		
        \node [inner, circle, fill=white!50, draw=black, inner sep=1pt, right=2.5cm of x31] (x41) {\small $x^4_1$};
        \node [inner, circle, fill=white!50, draw=black, inner sep=1pt, below=0.5cm of x41] (x42) {\small $x^4_2$};
		\node [inner, circle, fill=white!50, draw=black, inner sep=1pt, below=0.5cm of x42] (x43) {\small $x^4_3$};
		\node [inner, circle, fill=white!50, draw=black, inner sep=1pt, below=0.5cm of x43] (x44) {\small $x^4_4$};
		
        \node [inner, circle, fill=white!50, draw=black, inner sep=1pt, below=3.0cm of x14] (x10) {\small $o^1$};
        \node [inner, circle, fill=white!50, draw=black, inner sep=1pt, below=0.5cm of x10] (x30) {\small $o^3$};
		\draw (x10) to [out=150,in=130] (x11);
		\draw (x10) to [out=150,in=130] (x12);
		\draw (x10) to [out=150,in=130] (x13);
		\draw (x10) to [out=150,in=130] (x14);
		\draw (x30) to [out=50,in=230] (x31);
		\draw (x30) to [out=50,in=230] (x32);
		\draw (x30) to [out=50,in=230] (x33);
		\draw (x30) to [out=50,in=230] (x34);
		
        \node [inner, circle, fill=white!50, draw=black, inner sep=1pt, below=3.0cm of x44] (x20) {\small $e^2$};
        \node [inner, circle, fill=white!50, draw=black, inner sep=1pt, below=0.5cm of x20] (x40) {\small $e^4$};]
		\draw (x20) to [out=150,in=-50] (x21);
		\draw (x20) to [out=150,in=-50] (x22);
		\draw (x20) to [out=150,in=-50] (x23);
		\draw (x20) to [out=150,in=-50] (x24);
		\draw (x40) to [out=30,in=50] (x41);
		\draw (x40) to [out=30,in=50] (x42);
		\draw (x40) to [out=30,in=50] (x43);
		\draw (x40) to [out=30,in=50] (x44);
		
		\draw[red] (x11) to [out=10,in=170] (x21);
		\draw[red] (x11) to [out=-10,in=190] (x21);
		\draw[red] (x12) to [out=10,in=170] (x22);
		\draw[red] (x12) to [out=-10,in=190] (x22);
		\draw[red] (x13) to [out=10,in=170] (x23);
		\draw[red] (x13) to [out=-10,in=190] (x23);
		\draw[red] (x14) to [out=10,in=170] (x24);
		\draw[red] (x14) to [out=-10,in=190] (x24);
		
		\draw[red] (x21) to [out=10,in=170] (x31);
		\draw[red] (x21) to [out=-10,in=190] (x31);
		\draw[red] (x22) to [out=10,in=170] (x32);
		\draw[red] (x22) to [out=-10,in=190] (x32);
		\draw[red] (x23) to [out=10,in=170] (x33);
		\draw[red] (x23) to [out=-10,in=190] (x33);
		\draw[red] (x24) to [out=10,in=170] (x34);
		\draw[red] (x24) to [out=-10,in=190] (x34);
		
		\draw[red] (x31) to [out=10,in=170] (x41);
		\draw[red] (x31) to [out=-10,in=190] (x41);
		\draw[red] (x32) to [out=10,in=170] (x42);
		\draw[red] (x32) to [out=-10,in=190] (x42);
		\draw[red] (x33) to [out=10,in=170] (x43);
		\draw[red] (x33) to [out=-10,in=190] (x43);
		\draw[red] (x34) to [out=10,in=170] (x44);
		\draw[red] (x34) to [out=-10,in=190] (x44);
		
		\node [inner, circle, fill=white!50, draw=black, inner sep=1pt, below=3.5cm of x24] (z1) {\small $z_1$};
		\node [inner, circle, fill=white!50, draw=black, inner sep=1pt, right=2.5cm of z1] (z2) {\small $z_2$};
		\draw (z1) to (z2);
		\draw[red] (x10) to [out=10,in=170] (z1);
		\draw[red] (x10) to [out=-10,in=190] (z1);
		\draw[red] (x30) to [out=10,in=170] (z1);
		\draw[red] (x30) to [out=-10,in=190] (z1);
		
		\draw[red] (x20) to [out=170,in=10] (z2);
		\draw[red] (x20) to [out=190,in=-10] (z2);
		\draw[red] (x40) to [out=170,in=10] (z2);
		\draw[red] (x40) to [out=190,in=-10] (z2);
	\end{tikzpicture}
\caption{The graph $G$ for $n = 4$, $\Sigma = \{\sigma_1,\sigma_2,\sigma_3,\sigma_4\}$, and 
$A = (\Sigma \times \Sigma) \setminus \{(\sigma_1,\sigma_1), (\sigma_2,\sigma_2), (\sigma_3,\sigma_3), (\sigma_4,\sigma_4)\}$ (red double edges represent cages).}
\label{fig-graph}
\end{figure}

Given an instance $(W = (\Sigma, A), w_s, w_t)$ of the {\sc Word Reconfiguration} problem (where $|w_s| = |w_t| = n$), we assume, 
without loss of generality, that $n$ is even. 
We construct a graph $G$ as follows. Let $\Sigma = \{\sigma_1, \sigma_2, \ldots, \sigma_{|\Sigma|}\}$. 
We first create $n$ sets of vertices, $X_1$, $X_2$, $\ldots$, $X_n$, each of size $|\Sigma|$. 
We set $X_i = \{x^i_1, x^i_2, \ldots, x^i_{|\Sigma|}\}$. 
For each $X_i$, $i \in [n]$ and $i$ even, we create a vertex $e^i$ and add all edges between $e^i$ and vertices in $X_i$. 
For each $X_i$, $i \in [n]$ and $i$ odd, we create a vertex $o^i$ and add all edges between $o^i$ and vertices in $X_i$. 
We let $U_{\textsf{odd}} = \{o^i \mid i \text{ is odd}\}$ and $U_{\textsf{even}} = \{e^i \mid i \text{ is even}\}$. 
For all $1 \leq i < n$, $1 \leq j,j' \leq |\Sigma|$, $x^{i}_{j} \in X_{i}$, and $x^{i+1}_{j'} \in X_{i+1}$, we add a cage between 
$x^{i}_{j}$ and $x^{i+1}_{j'}$ in $G$ if and only if $(\sigma_j,\sigma_{j'}) \not\in A$. That is, we 
add a cage between $x^{i}_{j}$ and $x^{i+1}_{j'}$ whenever $\sigma_j$ and $\sigma_{j'}$ cannot appear consecutively in a word. 
Finally, we add two vertices $z_1$ and $z_2$ connected by an edge and we add a cage between $z_1$ and 
each vertex in $U_{\textsf{odd}}$ and a cage between $z_2$ and each vertex in $U_{\textsf{even}}$. 
We call the (extended) cage between vertices $u$ and $v$ a $(u,v)$-(extended) cage. 
This concludes the construction of the graph $G$ (see Figure~\ref{fig-graph}). 
Note that $|V(G)| = |\Sigma|(n + 1) + 2 + 8[m(n-1) + n]$, where $m = |\Sigma^2 \setminus A|$. 

\begin{lemma}\label{le-bip}
The graph $G$ is bipartite. 
\end{lemma}

\begin{proof}
Let $L = X_1 \cup X_3 \cup \ldots \cup X_{n-1} \cup U_{\textsf{even}} \cup \{z_1\}$ and 
$R = X_2 \cup X_4 \cup \ldots \cup X_n \cup U_{\textsf{odd}} \cup \{z_2\}$. 
First, note that both graphs $G[L]$ and $G[R]$ are edgeless, and therefore $G[L \cup R]$ is bipartite.  
Moreover, for any two vertices $u \in L$ and $v \in R$, either $uv \in E(G)$, $uv \not\in E(G)$, or 
$u$ and $b$ are connected by a distinct cage. Applying Proposition~\ref{pr-remain-bip}, we know that 
$G$ must also be bipartite. 
\end{proof}

We now describe independent sets of $G$ and distinguish between different types of tokens to help simplify the presentation. 
We will always have independent sets of size exactly $k = n + 1 + 3[m(n-1) + n]$. 
We say an independent set $I$ of $G$ is \emph{well-formed} if it has size $k$ and we can label its tokens as follows. 
One token, labeled the \emph{switch token}, is either on vertex $z_1$ or vertex $z_2$. 
For each $i \in [n]$, there is a token, labeled the \emph{$\Sigma_i$-token}, on some vertex in $X_i \cup \{e^i\}$ if $i$ is even and 
on some vertex in $X_i \cup \{o^i\}$ if $i$ is odd. 
We call the collection of all such tokens the \emph{$\Sigma$-tokens} (there are exactly $n$ $\Sigma$-tokens). 
For each $(u,v)$-extended cage $H$, there are three tokens, labeled \emph{$(u,v)$-caged tokens}, on vertices in $V(H)$ (which includes $u$ and $v$). 
The total number of caged tokens is $3[m(n-1) + n]$. 
We say $I$ is \emph{strictly well-formed} if it is well-formed and each $\Sigma_i$-token is on some vertex in $X_i$ (not on vertex $e^i$ or $o^i$). 

\begin{lemma}\label{le-well}
If $I$ is a well-formed independent set then there exists no $(u,v)$-cage in $G$ such that $\{u,v\} \subset I$.  
\end{lemma}

\begin{proof}
Assume otherwise and let $H$ be such an extended cage (including $u$ and $v$). 
Since $I$ is well-formed, we know that $|I \cap V(H)| \geq 3$. 
However, all vertices of $H$ are at distance at most one from either $u$ or $v$. 
Hence, if $|I \cap V(H)| \geq 3$ it cannot be the case that $\{u,v\} \subset I$. 
\end{proof}

\begin{lemma}\label{le-strict}
Let $I$ be a strictly well-formed independent set and let $x^1_{j_1}$, $x^2_{j_2}$, $\ldots$, and $x^n_{j_n}$ 
be the set of $\Sigma$-tokens. Then $w_I = \sigma_{j_1}\sigma_{j_2}\ldots\sigma_{j_n}$ is a word.  
\end{lemma}

\begin{proof}
It suffices to show that for any two consecutive symbols $\sigma_j$ and $\sigma_{j'}$ in $w_I$, we have $(\sigma_j,\sigma_{j'}) \in A$. 
Assume otherwise. Then, by our construction of $G$, there exists an ($x^i_{j}$,$x^{i + 1}_{j'}$)-cage 
and $\{x^i_{j},x^{i + 1}_{j'}\} \subset I$. Since $I$ is well-formed, we can apply Lemma~\ref{le-well} and get a contradiction.  
\end{proof}

\noindent 
Given a word $w$ and the graph $G$, we associate $w$ with a strictly well-formed independent set $I_w$ as follows. 
We add the switch token on vertex $z_1$ (or $z_2$). For each $i \in [n]$, we add a $\Sigma_i$-token on vertex $x^i_j$ whenever $w[i] = \sigma_j$, where $1 \leq j \leq |\Sigma|$ 
and $w[i]$ denotes the $i$th symbol of $w$. For each $(u,v)$-cage, $u,v \in V(G)$, we add three tokens 
on vertices $f_2$, $f_5$, and $f_7$ if there is no token on $v$ and we add three tokens on $f_3$, $f_6$, and $f_8$ otherwise. 

\begin{lemma}\label{le-forward}
If $(W = (\Sigma, A), w_s, w_t)$ is a yes-instance then $(G,I_{w_s},I_{w_t},k)$ is a yes-instance. 
\end{lemma}

\begin{proof}
Let $\langle w_1 = w_s, w_2, \ldots, w_{\ell-1}, w_\ell = w_t \rangle$ denote the corresponding sequence of words. 
Let $I_w$ and $I_{w'}$ be (arbitrarily chosen) strictly well-formed independent sets corresponding to two consecutive words in this sequence. 
It is enough to show that $\textsf{dist}_{\mc{S}_k(G)}(I_w,I_{w'})$ is finite. 
Recall that every two consecutive words differ in either odd or even positions. 
Therefore, to get from $I_w$ to $I_{w'}$, the following sequence of token slides will do. 
Assume that $w$ and $w'$ differ in odd positions. First, for each $(e,z_2)$-cage, we slide 
the token on $f_3$ to $f_2$, the token on $f_6$ to $f_5$, and the token on $f_8$ to $f_7$ (if needed). 
Now we can slide the switch token from $z_1$ to $z_2$. 
Then, for each $(o,z_1)$-cage, we slide the token on $f_3$ to $f_2$, the token on $f_6$ to $f_5$, and the token on $f_8$ to $f_7$. 
Then, we slide each $\Sigma_i$-token, for $i$ odd, to vertex $o^i$. 
For each ($x^i_{j}$,$x^{i + 1}_{j'}$)-cage, where $i$ is odd and $w'[i] = \sigma_j$ (and hence there is no token on $x^{i + 1}_{j'}$), we slide the token on 
$f_3$ to $f_2$, the token on $f_6$ to $f_5$, and the token on $f_8$ to $f_7$ (if needed). 
For each ($x^i_{j}$,$x^{i + 1}_{j'}$)-cage, where $i$ is odd, $w'[i] \neq \sigma_j$, and there is no token on $x^{i + 1}_{j'}$, we slide the token on 
$f_2$ to $f_3$, the token on $f_5$ to $f_6$, and the token on $f_7$ to $f_8$ (if needed). 
Next, for each odd $i \in [n]$ and assuming $w'[i] = \sigma_j$, we slide 
the token on $o^i$ to vertex $x^i_j$. The even case is handled similarly. 
\end{proof}

\begin{lemma}\label{le-struct}
Let $I$ be an independent set of $G$ and let $J$ be an independent set obtained from $I$ by sliding a single token. 
If $I$ is well-formed then so is $J$.   
\end{lemma}

\begin{proof}
We divide the proof into several cases. First, assume that the switch token slides 
into vertex $f_4$ of an $(o,z_1)$-cage $H$ (or symmetrically into vertex $f_1$ of an $(e,z_2)$-cage). 
Before the slide, the token must be on vertex $z_1$. If $I$ is well-formed then $|I \cap V(H)| \geq 3$ 
and none of those three tokens can be on $f_2$, $f_5$, or $f_7$. Hence, at least one 
token must be on either $f_3$, $f_6$, or $f_8$, implying that the switch token cannot slide from $z_1$ to $f_4$. 

Assume that a $\Sigma_i$-token, $i$ even (or symmetrically $i$ odd), slides into an $(e,z_2)$-cage ($(o,z_1)$-cage) $H$. 
When $I$ is well-formed we know that the switch token is either on $z_1$ or $z_2$.  
Before the slide, the $\Sigma_i$-token must be on vertex $e$. If the switch token is on $z_2$ then this is 
not possible (Lemma~\ref{le-well}). Hence, the switch token must be on $z_1$ and the $\Sigma_i$-token must slide from vertex $e$ 
to vertex $f_4$. But using the same argument as in the previous case we know that at least one 
token must be on either $f_3$, $f_6$, or $f_8$, implying that the $\Sigma_i$-token cannot slide from $e$ to $f_4$. 
The same arguments also hold for the case where a $\Sigma_i$-token slides into an ($x^i_{j}$,$x^{i + 1}_{j'}$)-cage. 

Assume that a ($x^i_{j}$,$x^{i + 1}_{j'}$)-caged token, $i$ odd, slides out of its extended cage, i.e., either slides to  
vertex $o^i \in U_{\textsf{odd}}$ or vertex $e^{i+1} \in U_{\textsf{even}}$. 
Since $I$ is well-formed, we know that, for each odd $i$, $X_i \cup \{o^i\}$ contains one $\Sigma_i$-token.  
In addition, for each even $i$, $X_i \cup \{e^i\}$ contains one $\Sigma_i$-token. 
Since $o^i$ is connected to all vertices in $X_i$ and $e^{i+1}$ is connected to all vertices in $X_{i+1}$, 
no ($x^i_{j}$,$x^{i + 1}_{j'}$)-caged token can slide to $o^i$ nor $e^{i+1}$. 

Finally, assume that an $(e,z_2)$-caged token or an $(o,z_1)$-caged token slides out of its extended cage. 
Since the switch token must be on $z_1$ or $z_2$, it must be the case that the caged token slides from 
some vertex $o^i \in U_{\textsf{odd}}$ or some vertex $e^i \in U_{\textsf{even}}$ to its neighbor in $X_i$. 
But this contradicts the fact that the $\Sigma_i$-token must be on some vertices in $X_i$.  
\end{proof}

\begin{lemma}\label{le-switch}
Let $\langle I_1, I_2, \ldots, I_{\ell-1}, I_\ell \rangle$ denote a reconfiguration sequence 
between two strictly well-formed independent sets in $\mc{S}_k(G)$.  
Let $I_p$ and $I_q$ be two sets in this sequence such that the following holds: 
\begin{itemize}
\item $1 \leq p < q \leq \ell$; 
\item $I_p$ is obtained from $I_{p-1}$ by sliding the switch token from $z_1$ to $z_2$ (or vice-versa); 
\item $I_q$ is obtained from $I_{q-1}$ by sliding the switch token from $z_1$ to $z_2$ (or vice-versa); 
\item There exists no set $I_r$ such that $p \le r \le q$ and $I_r$ is obtained from $I_{r - 1}$ by sliding the switch token. 
\end{itemize}
Then, $I_p$ and $I_q$ are strictly well-formed. Moreover, either $w_{I_p}[i] = w_{I_q}[i]$, for all even $i$, 
or $w_{I_p}[i] = w_{I_q}[i]$, for all odd $i$. 
\end{lemma}

\begin{proof}
Since $I_1$ is well-formed, we can assume, by Lemma~\ref{le-struct}, that all sets in the sequence are well-formed. 
To show that $I_p$ and $I_q$ are strictly well-formed we will in fact show that any independent set $I$ (in the sequence)  
obtained from its predecessor $I'$ by sliding the switch token must be strictly well-formed.  
Since $I$ is well-formed, it is enough to show that there are no $\Sigma$-tokens from $I$ on 
$U_{\textsf{odd}} \cup U_{\textsf{even}}$. Assume otherwise and consider the case where the switch token 
slides from $z_1$ to $z_2$ (the other case is symmetric). Applying Lemma~\ref{le-well}, we know that, in $I$, there can 
be no tokens on vertices in $U_{\textsf{odd}}$. Similarly, and since the switch token is the only token to slide, we 
know that in $I'$ there can be no tokens on vertices in $U_{\textsf{even}}$. Therefore, we can conclude that $I$, 
$I_p$, and $I_q$ are strictly well-formed. Moreover, by Lemma~\ref{le-strict}, we know that $w_{I}$, $w_{I_p}$, and $w_{I_q}$ are words. 

Now assume that $I_p$ is obtained from $I_{p-1}$ by sliding the switch token from $z_1$ to $z_2$ and 
$I_q$ is obtained from $I_{q-1}$ by sliding the switch token from $z_2$ to $z_1$. Then, 
for all $I_{p'}$, $p < p' < q$, $I_{p'} \cap U_{\textsf{even}} = \emptyset$. In other words, 
all $\Sigma_i$-tokens, $i$ even, cannot slide to vertex $e^i$ (Lemma~\ref{le-well}). 
Hence, $w_{I_p}[i] = w_{I_q}[i]$, for all even $i$. Considering the symmetric case 
where the switch token starts by sliding from $z_2$ to $z_1$ we get $w_{I_p}[i] = w_{I_q}[i]$, for all odd $i$. 
\end{proof} 

\begin{lemma}\label{le-backward}
If $(G,I_{w_s},I_{w_t},k)$ is a yes-instance then $(W = (\Sigma, A), w_s, w_t)$ is a yes-instance. 
\end{lemma}

\begin{proof}
Let $\eta = \langle I_1 = I_{w_s}, I_2, \ldots, I_{\ell-1}, I_\ell = I_{w_t}\rangle$ denote 
a reconfiguration sequence from $I_{w_s}$ to $I_{w_t}$ in $\mc{S}_k(G)$. 
Since $I_{w_s}$ is (strictly) well-formed, we can, from Lemma~\ref{le-struct}, assume 
that each set in the sequence is well-formed. 
Consider the subsequence $\eta' = \langle I_{p_1}, \ldots, I_{p_\ell} \rangle$ which is obtained by restricting $\eta$ 
to only those sets which are obtained from their predecessor by sliding the switch token. In other words, 
$I_{p_j}$ is obtained from $I_{p_j - 1}$ by sliding the switch token. 
We claim that $\langle w_s, w_{I_{p_1}}, \ldots, w_{I_{p_\ell}}, w_t \rangle$ is the required solution to the instance $(W = (\Sigma, A), w_s, w_t)$. 
Applying Lemma~\ref{le-switch}, we know that each independent set in $\eta'$ is strictly well-formed 
and therefore corresponds to a word (Lemma~\ref{le-strict}). Moreover, from Lemma~\ref{le-switch}, 
we know that each two consecutive words differ in either odd or even positions only, as needed. 
Note that $\langle w_s, w_{I_{p_1}}, \ldots, w_{I_{p_\ell}}, w_t \rangle$ might contain duplicate words which can safely be deleted. 
Moreover, if $\eta'$ is empty then $w_s$ and $w_t$ differ in only odd or even positions. 
\end{proof}

Theorem~\ref{th-pspace} follows by combining Proposition~\ref{pr-even-odd} and Lemmas~\ref{le-bip},~\ref{le-forward}, and~\ref{le-backward}. 

\begin{theorem}\label{th-pspace}
{\sc Token Sliding} is \PSPACE-complete on bipartite graphs. 
\end{theorem}

\bibliographystyle{siam}
\bibliography{references}

\begin{thebibliography}{10}

\bibitem{DBLP:journals/mst/Abu-KhzamFLS07}
{\sc F.~N. Abu{-}Khzam, M.~R. Fellows, M.~A. Langston, and W.~H. Suters}, {\em
  Crown structures for vertex cover kernelization}, Theory Comput. Syst., 41
  (2007), pp.~411--430.

\bibitem{Arnborg:1987:CFE:37170.37183}
{\sc S.~Arnborg, D.~G. Corneil, and A.~Proskurowski}, {\em Complexity of
  finding embeddings in a k-tree}, SIAM J. Algebraic Discrete Methods, 8
  (1987), pp.~277--284.

\bibitem{DBLP:journals/actaC/Bodlaender93}
{\sc H.~L. Bodlaender}, {\em A tourist guide through treewidth}, Acta Cybern.,
  11 (1993), pp.~1--21.

\bibitem{DBLP:journals/dam/BodlaenderT97}
{\sc H.~L. Bodlaender and D.~M. Thilikos}, {\em Treewidth for graphs with small
  chordality}, Discrete Applied Mathematics, 79 (1997), pp.~45--61.

\bibitem{DBLP:journals/corr/BonamyB16}
{\sc M.~Bonamy and N.~Bousquet}, {\em Token sliding on chordal graphs}, CoRR,
  abs/1605.00442 (2016).

\bibitem{DBLP:conf/swat/BonsmaKW14}
{\sc P.~S. Bonsma, M.~Kaminski, and M.~Wrochna}, {\em Reconfiguring independent
  sets in claw-free graphs}, in Algorithm Theory - {SWAT} 2014 - 14th
  Scandinavian Symposium and Workshops, Copenhagen, Denmark, July 2-4, 2014.
  Proceedings, 2014, pp.~86--97.

\bibitem{DBLP:journals/tcs/BrewsterMMN16}
{\sc R.~C. Brewster, S.~McGuinness, B.~Moore, and J.~A. Noel}, {\em A dichotomy
  theorem for circular colouring reconfiguration}, Theor. Comput. Sci., 639
  (2016), pp.~1--13.

\bibitem{CHJ08}
{\sc L.~Cereceda, J.~van~den Heuvel, and M.~Johnson}, {\em Connectedness of the
  graph of vertex-colourings}, Discrete Mathematics, 308 (2008), pp.~913--919.

\bibitem{pcbook}
{\sc M.~Cygan, F.~V. Fomin, L.~Kowalik, D.~Lokshtanov, D.~Marx, M.~Pilipczuk,
  M.~Pilipczuk, and S.~Saurabh}, {\em Parameterized Algorithms}, Springer,
  2015.

\bibitem{DBLP:conf/isaac/DemaineDFHIOOUY14}
{\sc E.~D. Demaine, M.~L. Demaine, E.~Fox{-}Epstein, D.~A. Hoang, T.~Ito,
  H.~Ono, Y.~Otachi, R.~Uehara, and T.~Yamada}, {\em Polynomial-time algorithm
  for sliding tokens on trees}, in Algorithms and Computation - 25th
  International Symposium, {ISAAC} 2014, Jeonju, Korea, December 15-17, 2014,
  Proceedings, H.~Ahn and C.~Shin, eds., vol.~8889 of Lecture Notes in Computer
  Science, Springer, 2014, pp.~389--400.

\bibitem{DBLP:books/daglib/0019278}
{\sc E.~D. Demaine and J.~O'Rourke}, {\em Geometric folding algorithms -
  linkages, origami, polyhedra}, Cambridge University Press, 2007.

\bibitem{DBLP:conf/isaac/Fox-EpsteinHOU15}
{\sc E.~Fox{-}Epstein, D.~A. Hoang, Y.~Otachi, and R.~Uehara}, {\em Sliding
  token on bipartite permutation graphs}, in Algorithms and Computation - 26th
  International Symposium, {ISAAC} 2015, Nagoya, Japan, December 9-11, 2015,
  Proceedings, 2015, pp.~237--247.

\bibitem{DBLP:conf/icalp/GharibianS15}
{\sc S.~Gharibian and J.~Sikora}, {\em Ground state connectivity of local
  hamiltonians}, in Automata, Languages, and Programming - 42nd International
  Colloquium, {ICALP} 2015, Kyoto, Japan, July 6-10, 2015, Proceedings, Part
  {I}, 2015, pp.~617--628.

\bibitem{GKMP09}
{\sc P.~Gopalan, P.~G. Kolaitis, E.~N. Maneva, and C.~H. Papadimitriou}, {\em
  The connectivity of {B}oolean satisfiability: computational and structural
  dichotomies}, SIAM Journal on Computing, 38 (2009), pp.~2330--2355.

\bibitem{Habib1994}
{\sc M.~Habib and R.~H. M{\"o}hring}, {\em Treewidth of cocomparability graphs
  and a new order-theoretic parameter}, Order, 11 (1994), pp.~47--60.

\bibitem{H87}
{\sc P.~Hall}, {\em On representatives of subsets}, in Classic Papers in
  Combinatorics, Birkhäuser Boston, 1987, pp.~58--62.

\bibitem{DBLP:journals/tcs/HearnD05}
{\sc R.~A. Hearn and E.~D. Demaine}, {\em Pspace-completeness of sliding-block
  puzzles and other problems through the nondeterministic constraint logic
  model of computation}, Theor. Comput. Sci., 343 (2005), pp.~72--96.

\bibitem{presentation1}
{\sc D.~A. Hoang and R.~Uehara}, {\em Sliding tokens on a cactus}.
\newblock \url{https://hoanganhduc.github.io/events/ISAAC2016/slide.pdf}, 2016.
\newblock [ISAAC Presentation slides].

\bibitem{DBLP:journals/corr/HoangU17}
\leavevmode\vrule height 2pt depth -1.6pt width 23pt, {\em Polynomial-time
  algorithms for sliding tokens on cactus graphs and block graphs}, CoRR,
  abs/1705.00429 (2017).

\bibitem{IDHPSUU11}
{\sc T.~Ito, E.~D. Demaine, N.~J.~A. Harvey, C.~H. Papadimitriou, M.~Sideri,
  R.~Uehara, and Y.~Uno}, {\em On the complexity of reconfiguration problems},
  Theoretical Computer Science, 412 (2011), pp.~1054--1065.

\bibitem{IKD12}
{\sc T.~Ito, M.~Kami\'{n}ski, and E.~D. Demaine}, {\em Reconfiguration of list
  edge-colorings in a graph}, Discrete Applied Mathematics, 160 (2012),
  pp.~2199--2207.

\bibitem{DBLP:conf/tamc/ItoKOSUY14}
{\sc T.~Ito, M.~Kaminski, H.~Ono, A.~Suzuki, R.~Uehara, and K.~Yamanaka}, {\em
  On the parameterized complexity for token jumping on graphs}, in Theory and
  Applications of Models of Computation - 11th Annual Conference, {TAMC} 2014,
  Chennai, India, April 11-13, 2014. Proceedings, 2014, pp.~341--351.

\bibitem{DBLP:journals/ieicet/ItoNZ16}
{\sc T.~Ito, H.~Nooka, and X.~Zhou}, {\em Reconfiguration of vertex covers in a
  graph}, {IEICE} Transactions, 99-D (2016), pp.~598--606.

\bibitem{JS79}
{\sc W.~W. Johnson and W.~E. Story}, {\em Notes on the ``15'' puzzle}, American
  Journal of Mathematics, 2 (1879), pp.~397--404.

\bibitem{KMM12}
{\sc M.~Kami\'{n}ski, P.~Medvedev, and M.~Milani\v{c}}, {\em Complexity of
  independent set reconfigurability problems}, Theoretical Computer Science,
  439 (2012), pp.~9--15.

\bibitem{DBLP:conf/stacs/KanjX15}
{\sc I.~A. Kanj and G.~Xia}, {\em Flip distance is in {FPT} time o(n+ k *
  c{\^{}}k)}, in 32nd International Symposium on Theoretical Aspects of
  Computer Science, {STACS} 2015, March 4-7, 2015, Garching, Germany, 2015,
  pp.~500--512.

\bibitem{KPS08}
{\sc G.~Kendall, A.~J. Parkes, and K.~Spoerer}, {\em A survey of {NP}-complete
  puzzles}, ICGA Journal,  (2008), pp.~13--34.

\bibitem{DBLP:conf/wads/LokshtanovMPRS15}
{\sc D.~Lokshtanov, A.~E. Mouawad, F.~Panolan, M.~S. Ramanujan, and
  S.~Saurabh}, {\em Reconfiguration on sparse graphs}, in Algorithms and Data
  Structures - 14th International Symposium, {WADS} 2015, Victoria, BC, Canada,
  August 5-7, 2015. Proceedings, 2015, pp.~506--517.

\bibitem{DBLP:journals/comgeo/LubiwP15}
{\sc A.~Lubiw and V.~Pathak}, {\em Flip distance between two triangulations of
  a point set is np-complete}, Comput. Geom., 49 (2015), pp.~17--23.

\bibitem{DBLP:conf/icalp/MouawadNPR15}
{\sc A.~E. Mouawad, N.~Nishimura, V.~Pathak, and V.~Raman}, {\em Shortest
  reconfiguration paths in the solution space of boolean formulas}, in
  Automata, Languages, and Programming - 42nd International Colloquium, {ICALP}
  2015, Kyoto, Japan, July 6-10, 2015, Proceedings, Part {I}, 2015,
  pp.~985--996.

\bibitem{DBLP:conf/isaac/MouawadNR14}
{\sc A.~E. Mouawad, N.~Nishimura, and V.~Raman}, {\em Vertex cover
  reconfiguration and beyond}, in Algorithms and Computation - 25th
  International Symposium, {ISAAC} 2014, Jeonju, Korea, December 15-17, 2014,
  Proceedings, 2014, pp.~452--463.

\bibitem{MouawadN0SS17}
{\sc A.~E. Mouawad, N.~Nishimura, V.~Raman, N.~Simjour, and A.~Suzuki}, {\em On
  the parameterized complexity of reconfiguration problems}, Algorithmica, 78
  (2017), pp.~274--297.

\bibitem{DBLP:journals/jct/RobertsonS84}
{\sc N.~Robertson and P.~D. Seymour}, {\em Graph minors. {III.} planar
  tree-width}, J. Comb. Theory, Ser. {B}, 36 (1984), pp.~49--64.

\bibitem{SeymourT93}
{\sc P.~D. Seymour and R.~Thomas}, {\em Graph searching and a min-max theorem
  for tree-width}, J. Comb. Theory, Ser. {B}, 58 (1993), pp.~22--33.

\bibitem{H13}
{\sc J.~van~den Heuvel}, {\em The complexity of change}, Surveys in
  Combinatorics 2013, 409 (2013), pp.~127--160.

\bibitem{DBLP:journals/corr/Wrochna14}
{\sc M.~Wrochna}, {\em Reconfiguration in bounded bandwidth and treedepth},
  CoRR, abs/1405.0847 (2014).

\bibitem{Wrochna15}
\leavevmode\vrule height 2pt depth -1.6pt width 23pt, {\em Homomorphism
  reconfiguration via homotopy}, in 32nd International Symposium on Theoretical
  Aspects of Computer Science, {STACS} 2015, March 4-7, 2015, Garching,
  Germany, 2015, pp.~730--742.

\bibitem{doi:10.1137/0602010}
{\sc M.~Yannakakis}, {\em Computing the minimum fill-in is {NP}-complete}, SIAM
  Journal on Algebraic Discrete Methods, 2 (1981), pp.~77--79.

\end{thebibliography}
\end{document}